\documentclass[11pt]{article}
\usepackage[margin=1.45in]{geometry}

\usepackage{mdwlist}
\usepackage{enumerate}

\usepackage{amsmath,amssymb}
\usepackage{graphics, subfigure, float}
\usepackage{fp, calc}
\usepackage{hyperref}
\usepackage{url}

\usepackage[T1]{fontenc} % IMPORTANT: Load before the fonts
\usepackage{fourier}
\usepackage{bm}

\usepackage{amscd,amsthm}

\usepackage{pst-all}
\usepackage{pstricks-add}
\usepackage{pst-func}
\newpsobject{showgrid}{psgrid}{subgriddiv=1,griddots=10,gridlabels=6pt}
\usepackage{verbatim, comment}
\usepackage{datetime}

\newtheoremstyle{theorem}{1em}{1em}{\slshape}{0pt}{\bfseries}{.}{ }{}
\theoremstyle{theorem}
\newtheorem{theorem}{Theorem}

\newtheorem{conjecture}{Conjecture}
\newtheorem*{theorem*}{Theorem}
\newtheorem{corollary}[theorem]{Corollary}

\newtheorem{lemma}[theorem]{Lemma}

\newtheorem*{claim*}{Claim}

\theoremstyle{remark}
\newtheorem{remark}{Remark}
\newtheorem*{remark*}{Remark}

\newtheoremstyle{example}{1em}{1em}{}{0pt}{\bfseries}{.}{ }{}

\providecommand{\setR}{\mathbb{R}}

\newcommand{\vol}{\mathrm{vol}}

\newcommand{\conv}{\textrm{conv}}

\newcommand{\inn}[2]{\langle {#1}, {#2} \rangle}
\newcommand{\E}{\mathop{\mathbb{E}}}

\newcommand{\eps}{\varepsilon}

\psset{linewidth=1pt, arrowsize=6pt}

\usepackage{calrsfs} % now e.g. \pazocal{I} gives a nice I
\DeclareMathAlphabet{\pazocal}{OMS}{zplm}{m}{n}

\usepackage[displaymath,textmath,graphics, subfigure, floats]{preview} %sections,
\PreviewEnvironment{center} 

\makeatother

\title{Vector Balancing in Lebesgue Spaces}
\date{}
\author{Victor Reis \thanks{University of Washington, Seattle. Email: {\tt voreis@uw.edu}.} \and Thomas Rothvoss  \thanks{University of Washington, Seattle. Email: {\tt rothvoss@uw.edu}. 
Supported by NSF CAREER grant 1651861 and a David \& Lucile Packard Foundation Fellowship.}}

\begin{document}
\maketitle
\begin{abstract}

The \emph{Koml\'os} conjecture suggests that for any vectors  $\bm{a}_1,\ldots,\bm{a}_n \in B_2^m$ there exist $x_1, \dots, x_n \in \{ -1,1\}$ so that $\|\sum_{i=1}^n x_i\bm{a}_i\|_\infty \le O(1)$. It is a natural extension to ask what $\ell_q$-norm bound to expect for $\bm{a}_1,\ldots,\bm{a}_n \in B_p^m$. We prove a tight partial coloring result for such vectors, implying a nearly tight full coloring bound. As a corollary, this implies a special case of Beck-Fiala's conjecture. We achieve this by showing that, for any $\delta > 0$, a symmetric convex body $K \subseteq \setR^n$ with Gaussian measure at least $e^{-\delta n}$ admits a partial coloring. Previously this was known only for a \emph{small} enough $\delta$. Additionally, we show that a hereditary volume bound suffices to provide such Gaussian measure lower bounds.

\end{abstract}

\newpage

\section{Introduction}

The celebrated \emph{Spencer's Theorem} in discrepancy theory \cite{Spencer1985SixSD} shows that "six standard deviations suffice" for balancing vectors in the $\ell_\infty$-norm: for any $\bm{a}_1, \dots, \bm{a}_n \in [-1,1]^n$, there exist signs $\bm{x} \in \{-1,1\}^n$ such that $\|\sum_{i=1}^n x_i\bm{a}_i\|_\infty \le 6\sqrt{n}$. More generally, Spencer showed that for vectors in $[-1,1]^m$ with $n\le m$ one can achieve a bound of $O(\sqrt{n \log(2m/n)})$. While his proof used a nonconstructive form of the \emph{partial coloring lemma} based on the pigeonhole principle, in the past decade several approaches starting with the breakthrough work of Bansal \cite{DiscrepancyMinimization-Bansal-FOCS2010} did succeed in computing such signs in polynomial time \cite{DiscrepancyMinimization-LovettMekaFOCS12, ConstructiveDiscrepancy-Rothvoss-FOCS2014, DBLP:journals/corr/LevyRR16, DBLP:journals/rsa/EldanS18}.

As for balancing vectors of bounded $\ell_2$-norm, the situation has been more delicate. In the same paper, Spencer \cite{Spencer1985SixSD} showed a nonconstructive bound of $O(\log n)$ for the $\ell_\infty$ discrepancy of vectors $\bm{a}_1, \dots, \bm{a}_n \in B^m_2$ and also stated a discrete version of a conjecture of Koml\'os that this may be improved to $O(1)$. This was improved to $O(\sqrt{\log n})$ by Banaszczyk \cite{BalancingVectors-Banaszczyk98} who showed that in fact for any set of $n$ vectors of $\ell_2$-norm at most 1 and any convex body $K \subseteq \setR^m$ of Gaussian measure at least $1/2$, some $\pm 1$ combination of such vectors lies in $5\cdot K$. For the general setting of $\ell_q$ discrepancy, Matou\v{s}ek ~\cite{Matousek98} gave an upper bound of $O(q) \cdot m^{1/q}$ for balancing vectors from $\ell_2$ to $\ell_q$. More recently, the work of Barthe, Gu\'edon, Mendelson and Naor \cite{GeometryOfLpBall-BartheGuedonMendelsonNaor2005} (see Prop. 25) shows that, for  $q \ge 2$, $n$-dimensional slices of the $\ell_q$ ball in $\setR^m$ scaled by a factor of  $O(\sqrt{q}) \cdot n^{1/q}$ do have Gaussian measure at least $1/2$ (we include an alternate proof in the appendix), thus improving the bound to $O(\sqrt{q}) \cdot  n^{1/q}$. For $q = \log n$, this matches the $\ell_2$ to $\ell_\infty$ bound of $O(\sqrt{\log n})$. Banaszczyk's proof was nonconstructive and the first polynomial time algorithm in the general convex body setting was found only recently by Bansal, Dadush, Garg and Lovett \cite{GramSchmidtWalk-BansalDGL-STOC18}, while the Koml\'os conjecture remains an open problem.
The work of \cite{GramSchmidtWalk-BansalDGL-STOC18} actually shows that for any vectors $\bm{a}_1,\ldots,\bm{a}_n \in B_2^m$
there exists an efficiently computable distribution over signs $\bm{x} \in \{ -1,1\}^n$ so that the sum $\bm{X} := \sum_{i=1}^n x_i\bm{a}_i$
is $O(1)$-subgaussian, meaning that $\E[e^{\inn{ \bm{\theta}} {\bm{X}}}] \le e^{O(1) \|\bm{\theta}\|_2^2}$ for every $\bm{\theta} \in \setR^m$, and will be in $O(1) \cdot K$ with good probability. Interestingly, this means their
algorithm is \emph{oblivious} to the body $K$, which is a striking difference to the regime of $\gamma_n(K) = e^{-\Theta(n)}$ where any algorithm needs to be dependent on $K$. The connection between
Banaszczyk's theorem and subgaussianity is due to Dadush et al.~\cite{DBLP:conf/approx/DadushGLN16}.

For the general setting of balancing vectors from $\ell_p$ to $\ell_q$, where we are given vectors $\bm{a}_1, \dots, \bm{a}_n \in B^m_p$ and wish to find signs $x_1, \dots, x_n$ that minimize the $\ell_q$ norm of $\sum_{i=1}^n x_i \bm{a}_i$ (also called $\ell_q$ discrepancy), not much was known beyond Spencer's theorem ($p = \infty$) or what can be deduced from Banaszczyk's theorem as above: any vector in $B^m_p$ also belongs to $m^{\max(0,1/2-1/p)} \cdot B^m_2$, thus implying a discrepancy bound of $O(\sqrt{q}) \cdot m^{\max(0, 1/2-1/p)} \cdot n^{1/q}$. Even in the square case $m = n$, in spite of tight partial coloring bounds \cite{Spencer1985SixSD}, it has been an open problem to remove the dependency on $\sqrt{q}$ \cite{8555088}. The goal of this paper is to provide a unified approach for balancing from $\ell_p$ to $\ell_q$ via optimal constructive fractional partial colorings, which yield optimal bounds for most of the range $1 \le p \le q \le \infty$. We obtain such fractional partial colorings by proving a new measure lower bound on the relevant linear preimages of $\ell_q$ balls (Section 3) and an improved algorithm for sets of Gaussian measure $e^{-\delta n}$ for any $\delta > 0$ (Section 4), as opposed to previous work (\cite{ConstructiveDiscrepancy-Rothvoss-FOCS2014, DBLP:journals/rsa/EldanS18}) which required measure $e^{-\delta n}$ for \emph{sufficiently small} $\delta > 0$. Finally, we show that a \textit{hereditary} volume lower bound is sufficient to imply such Gaussian measure bound (Section 5). 

As an application, we show a slight improvement to the bounds for the well-known Beck-Fiala conjecture \cite{BECK19811}, a discrete version of Koml\'os. It asks for a $O(\sqrt{t})$ bound on the $\ell_\infty$ discrepancy of any $\bm{a}_1, \dots, \bm{a}_n \in \{0,1\}^m$, each with at most $t$ ones. We establish the conjecture for $t \ge n$ and show slightly improved bounds when $t$ is close to $n$ (Corollary \ref{cor:BeckFiala}).

\textbf{Notation.} Let $B^m_p := \{\bm{x} \in \setR^m : \|\bm{x}\|_p \le 1\}$ denote the unit ball in the $\ell_p$-norm. The \emph{Gaussian measure} of a measurable set $K \subseteq \setR^n$ is given by $\gamma_n(K) := \Pr_{\bm{x} \sim N(\bm{0}, \bm{I}_n)} [\bm{x} \in K]$. We denote the \emph{mean width} of a convex set as $w(K) := \E_{\bm{\theta} \in S^{n-1}}[\sup_{\bm{x} \in K} \left<\bm{\theta},\bm{x}\right>]$. The Euclidean distance to a set $S \subseteq \setR^n$ is denoted by $d(\bm{x},S) := \min\{ \|\bm{x}-\bm{y}\|_2 : \bm{y} \in S\}$. A function $f : \setR^m \to \setR$ is $\alpha$-Lipschitz if $|f(\bm{x}) - f(\bm{y})| \le \alpha \cdot \|\bm{x}-\bm{y}\|_2$ for $\bm{x}, \bm{y} \in \setR^m$. If $\bm{A} \in \setR^{m \times n}$ is a matrix, we denote its rows by $\bm{A}_1,\ldots,\bm{A}_m \in \setR^n$ and its columns by $\bm{a}_1,\ldots,\bm{a}_n \in \setR^m$. Naturally, a matrix can also be interpreted as a (not necessa{}rily invertible) linear map. Then for any set $K \subseteq \setR^m$, we use the notation $\bm{A}^{-1}(K) := \{ \bm{x} \in \setR^n : \bm{A}\bm{x} \in K\}$. The $C$-scaling of a symmetric convex body $K$ is the body $C \cdot K = \{c\bm{x} : \bm{x} \in K\}$. 

% \newpage
\subsection{Our contribution \label{sec:Contribution}}

Our main contribution is a tight bound on partial colorings for balancing from $\ell_p$ to $\ell_q$:

\begin{theorem} \label{thm:PartialColoringForLpLq}
Let $n \le m$ and $2 \le p \le q \le \infty$. \footnote{When $p \le 2$, uniformly random signs achieve a tight bound of $\Theta(n^{1/q})$ (see Theorem~\ref{thm:LowerBound}), so we focus on the more interesting case $p \ge 2$.} Then for any $\bm{a}_1, \dots, \bm{a}_n \in B^m_p$, there exists a polynomial-time computable partial coloring $\bm{x} \in [-1,1]^n$ with $|\{i : x_i^2 = 1\}| \ge n/2$ so that
\[
\displaystyle \Big\|\sum_{i=1}^n x_i \bm{a}_i\Big\|_q \le C \sqrt{\min\Big(p,\log\Big(\frac{2m}{n}\Big)\Big)} \cdot n^{1/2 - 1/p + 1/q},
\]
for some universal constant $C > 0$.
\end{theorem}

By a linear algebraic argument due to B\'ar\'any and Grinberg \cite{BARANY19811}, the condition $n \le m$ does not weaken the theorem: in fact for $n > m$ the upper bound can only be larger than that of $n = m$ by a factor of two. On the other hand, the condition $p \le q$ is natural, for otherwise if $p > q$ we would need a polynomial dependence on the dimension $m$, even for $n = 1$. By iteratively applying Theorem~\ref{thm:PartialColoringForLpLq} we can obtain a full coloring
at the expense of another factor of $\frac{1}{1/2 - 1/p + 1/q}$, with the caveat that $p > 2$ whenever $q = \infty$:

\begin{theorem} \label{thm:FullColoringForLpLq}
Let $n \le m$ and $2 \le p \le q \le \infty$ with $\{p,q\} \neq \{2, \infty\}$. Then for any $\bm{a}_1, \dots, \bm{a}_n \in B^m_p$, there exist polynomial-time computable signs $\bm{x} \in \{-1,1\}^n$ so that
\[
\displaystyle \Big\|\sum_{i=1}^n x_i \bm{a}_i\Big\|_q \le \frac{C \sqrt{\min\Big(p,\log\Big(\frac{2m}{n}\Big)\Big)} }{1/2 -1/p+1/q} \cdot n^{1/2 - 1/p + 1/q},
\]
for some universal constant $C > 0$.
\end{theorem}

This significantly improves upon the general $\sqrt{q} \cdot m^{1/2-1/p} \cdot n^{1/q}$ bound from Banaszczyk's theorem in  \cite{8555088}  when $p = 2 + \eps$ for (not too small) $\eps > 0$ and $q \gg 1$. It is also worth noting that we may always assume $q \le \log(m)$ as larger norms are equivalent up to a constant by Lemma~\ref{lem:ElementaryInequalityOnLpNorms}. When $p = q$ and $m=n$, we get the following corollary which matches, up to a constant, the lower bound $\Omega(\sqrt{n})$ of \cite{Banaszczyk1993} known to hold for any norm:

\begin{corollary}[$\ell_p$ version of Spencer's theorem] \label{cor:LpSpencer}
Let $2 \le p \le \infty$ and $n \in \mathbb{N}$. Then for any $\bm{a}_1, \dots, \bm{a}_n \in B^n_p$, there exist polynomial-time computable signs $\bm{x} \in \{-1,1\}^n$ so that
\[
\displaystyle \Big\|\sum_{i=1}^n x_i \bm{a}_i\Big\|_p \le C\sqrt{n},
\]
for some universal constant $C > 0$.
\end{corollary}

The following corollary shows the Beck-Fiala conjecture holds for $t \ge n$ and slightly improves upon the best known bound of $O(\sqrt{t \log n})$ \cite{BalancingVectors-Banaszczyk98} when $t$ is close to $n$:
%\newpage
\begin{corollary}[Bound for Beck-Fiala] \label{cor:BeckFiala} Let $n \le m$ and $\bm{a}_1, \dots, \bm{a}_n \in \{0,1\}^m$, each with at most $t \in [m]$ ones. Then there exist polynomial-time computable signs $\bm{x} \in \{-1,1\}^n$ so that

\[
\displaystyle \Big\|\sum_{i=1}^n x_i \bm{a}_i\Big\|_\infty \le C\sqrt{t} \log\Big(\frac{2\max(n, t)}{t}\Big),
\]
for some universal constant $C > 0$.
\end{corollary}

We show the partial coloring bound in Theorem~\ref{thm:PartialColoringForLpLq} is tight at least when $m = n$:
\begin{theorem} \label{thm:LowerBound}
Let $1 \le p \le q \le \infty$. There exist infinitely many positive integers $n$ for which we can find $\bm{a}_1, \dots, \bm{a}_n \in B^n_p$ such that for any $\bm{x} \in [-1,1]^n$ with $|\{i : x_i^2 = 1\}| \ge n/2$ one has
\[
\displaystyle \Big\|\sum_{i=1}^n x_i \bm{a}_i\Big\|_q \ge C \cdot n^{\max(0,1/2 - 1/p) + 1/q},
\]
for some universal constant $C > 0$.
\end{theorem}
%\rem{T: We are mentioning a partial coloring algorithm for any $K$ with $\gamma_n(K) \geq e^{-\alpha n}$ but don't state any in the contributions section. We might at least state the final one of Thm~\ref{thm:PartialColoringForExpSmallSetsWithShiftAndLargeFractionColored} here. That's the one that we use and it's awkward to use a Theorem before it even got stated. }

A result of Giannopoulos~\cite{Giannopoulos1997} shows that for a \emph{small enough} constant, a symmetric convex body $K$ with $\gamma_n(K) \geq e^{-\alpha n}$
contains a partial coloring $\bm{x} \in \{ -1,0,1\}^n$ with a linear number of entries in $\pm 1$. We can prove that for fractional colorings \emph{any} constant $\alpha > 0$ suffices. Our argument even works for intersections with a large enough subspace.
\begin{theorem} \label{thm:PartialColoringForExpSmallSetsWithShiftAndLargeFractionColored} %\rem{T: Needed to get $n/2$ elements colored and not just $\Omega(n)$.}
For all $\alpha,\beta,\gamma>0$, there is a constant $C := C(\alpha,\beta,\gamma)>0$ so that the following holds: There is a randomized polynomial time algorithm which for a symmetric convex set $K \subseteq \setR^n$ with $\gamma_n(K) \geq e^{-\alpha n}$, a shift $\bm{y} \in [-1,1]^n$ and a subspace $H \subseteq \setR^n$ with $\dim(H) \geq \beta n$, finds an  $\bm{x} \in (C \cdot K \cap H)$ with $\bm{x} + \bm{y} \in [-1,1]^n$ and $|\{i \in [n] : (\bm{x} + \bm{y})_i \in \{ \pm 1\} \}| \geq (\beta-\gamma) n$.
\end{theorem}

Finally, we show that a weaker \textit{hereditary} volume lower bound suffices to provide Gaussian measure lower bounds for arbitrary convex bodies. Previously such an implication was known only for the Gaussian measure of intersections with subspaces ~\cite{8555088}:

%\newpage 

\begin{theorem}\label{thm:HereditaryVolumeImpliesGaussianMeasure}
Let $K \subseteq \setR^n$ be a symmetric convex body. Given $S \subseteq [n]$, denote by $K_S$ the intersection with the coordinate subspace: $K_S := K \cap \{\bm{x} : x_i = 0 \ \forall i \notin S\} \subseteq \setR^S$. Then we have

\[
\gamma_n(K) \ge \min_{S \subseteq [n]} \vol_{|S|}(K_S) \cdot 2^{-O(n)},\] 
with the convention that $\vol_{0} (\{\bm{0}\}) = 1$. More generally, for any $\delta \in (0, 1]$,
\[
\gamma_n(K) \ge \min_{S \subseteq [n], |S| \le \delta n} \vol_{|S|}(K_S)^{1/\delta} \cdot 2^{-O(n/\delta)}.\] 
\end{theorem}

\section{Preliminaries}

We will use two elementary inequalities dealing with $\ell_p$-norms. The first one estimates the ratio between different norms: 
\begin{lemma} \label{lem:ElementaryInequalityOnLpNorms}
For any $\bm{z} \in \setR^m$ and $1 \leq p \leq q \leq \infty$, we have $\|\bm{z}\|_q \leq \|\bm{z}\|_p \leq m^{1/p - 1/q}\|\bm{z}\|_q$.
\end{lemma}
It is instructive to note that this bound implies $\|\bm{z}\|_{\infty} \leq \|\bm{z}\|_{\log_2(m)} \leq 2\|\bm{z}\|_{\infty}$. 
 If one has an upper bound on the largest entry in a vector --- say $\|\bm{z}\|_{\infty} \leq 1$ --- then one can
 strengthen the first inequality to $\|\bm{z}\|_q^q \leq \|\bm{z}\|_p^p$. More generally: 
 \begin{lemma} \label{lem:InterpolationInequality}
For any $\bm{z} \in \setR^m$ and $1 \leq p \leq q \leq \infty$, we have  $\|\bm{z}\|_q^q \le \|\bm{z}\|_p^p \cdot \|\bm{z}\|_\infty^{q-p}$.
\end{lemma}

We will also need the following version of \emph{Khintchine's inequality}, see e.g. the excellent textbook of
Artstein-Avidan, Giannopoulos and Milman~\cite{AsymptoticGeometricAnalysisBook2005}.
\begin{lemma}[Khintchine's inequality] \label{lem:Khintchine}
Given $p > 0$, $a_1, \dots, a_n \in \setR$ and $\bm{x} \sim N(\bm{0},\bm{I}_n)$, we have
\[
\E\Big[\Big|\sum_{i=1}^n x_i a_i\Big|^p \Big] \le C \sqrt{p} \cdot \Big(\sum_{i=1}^n a_i^2\Big)^{p/2} 
\]
where $C>0$ is a universal constant.
\end{lemma} 
This fact can be derived from a standard Chernov bound which guarantees that for a vector with $\|\bm{a}\|_2 = 1$ one has $\Pr[|\left<\bm{a},\bm{x}\right>| > \lambda] \leq 2e^{-\lambda^2/2}$; then one can analyze that the regime of $\lambda = \Theta(\sqrt{p})$ dominates the contribution to $\E[|\left<\bm{a},\bm{x}\right>|^p]$. We use it to show the following standard estimate on the type constants of $\ell_p$ spaces (see Appendix A):
\begin{lemma} \label{lem:ExpectedLpNorm}
Given $p \ge 1$ and $\bm{a}_1, \dots, \bm{a}_n \in B^m_p$ and $\bm{x} \sim N(\bm{0}, \bm{I}_n)$, we have
\[
\E\Big[\Big\|\sum_{i=1}^n x_i \bm{a}_i\Big\|_p \Big] \le O(\sqrt{p} \cdot n^{\max(1/2,1/p)}).
\]
\end{lemma}

A well-known correlation inequality for Gaussian measure is the following:
\begin{lemma}[\v{S}idak~\cite{SidaksLemma67} and Kathri~\cite{KhatriCorrelationInequality67}] \label{lem:SidakLemma}
For any symmetric convex set $K \subseteq \setR^n$ and strip $S =\{ \bm{x} \in \setR^n : |\left<\bm{a},\bm{x}\right>| \leq 1\}$, one has $\gamma_n(K \cap S) \geq \gamma_n(K) \cdot \gamma_n(S)$.
\end{lemma}

It is worth noting that a recent result of Royen~\cite{ProofOfGCI-Royen-Arxiv2014} extends this to any two arbitrary symmetric sets, though its full power will not be needed. We refer to the exposition of Lata{\l}a and Matlak~\cite{RoyensProofOfGCI-LatalaMatlak-Arxiv2017}. We also need a one-dimensional estimate:

\begin{lemma} \label{lem:OneDimEstimate}
For a strip $S =\{ \bm{x} \in \setR^n : |\left<\bm{a},\bm{x}\right>| \leq 1\}$, one has \[
\gamma_n(S) = \gamma_1 (\{x \in \setR: |x| \le \|\bm{a}\|^{-1}_2\}) \ge 1-\exp(-\|\bm{a}\|_2^{-2}/2).
\]
\end{lemma}

We use the following scaling lemma to deal with constant factors, see ~\cite{tkocz2015high}:

\begin{lemma} \label{lem:ScalingLemma}
 Let $K \subset \setR^n$ be a measurable set and $B$ be a closed Euclidean ball such that $\gamma_n(K) = \gamma_n(B)$. Then $\gamma_n(tK) \ge \gamma_n(tB)$ for all $t \in [0, 1]$. In particular, if $\gamma_n(C \cdot K) \ge 2^{-O(n)}$ for some constant $C > 1$ then also $\gamma_n(K) \ge 2^{-O(n)}$.  
\end{lemma}

For Section 4 we also need three helpful results. For the first one, see~\cite{Handel2014ProbabilityIH}.
\begin{theorem} \label{thm:ConcentrationFor1LipschitzFunctions}
If $F: \setR^m \to \setR$ is $1$-Lipschitz, then for $t \geq 0$ one has

\[
\Pr_{\bm{y} \sim N(\bm{0}, \bm{I}_m)} \big[F(\bm{y}) > \E[F(\bm{y})] + t\big] \le e^{-t^2/2}.
\]
\end{theorem}

%\begin{theorem}[Urysohn's Inequality~\cite{AsymptoticGeometricAnalysisBook2005}]
%For any convex body $K \subseteq \setR^n$ one has
%\[
%w(K) \ge \Big(\frac{\vol_n(K)}{\vol_n(B^n_2)}\Big)^{1/n}.
%\]
% \end{theorem}

The classical \emph{Urysohn Inequality} states that among all convex bodies of identical volume,
the Euclidean ball minimizes the width. We will need a variant that is phrased in terms of the Gaussian
measure rather than volume. For a proof, see Eldan and Singh~\cite{DBLP:journals/rsa/EldanS18}.
\begin{theorem}[Gaussian Variant of Urysohn's Inequality] \label{thm:UrysohnInequality}
  Let $K \subseteq \setR^n$ be a convex body and let $r > 0$ be so that $\gamma_n(K) = \gamma_n(r B_2^n)$. Then
  $w(K) \geq w(r B_2^n) = r$.
\end{theorem}

For a symmetric convex body $K$ and a subspace $H$, the Gaussian measure of the section $K \cap (\bm{x} + H)$ is maximized when $\bm{x} = \bm{0}$ by log-concavity. Thus we have the following:

\begin{lemma}[Gaussian measure of sections] \label{lem:GaussianMeasureSections}
        Let $K \subseteq \setR^n$ be a symmetric convex body and $H \subseteq \setR^n$ a subspace. Then $\gamma_H(K \cap H) \ge \gamma_n(K)$.
\end{lemma}

%We also make use of a result due to Ball and Pajor~\cite{ConvexBodiesWithFewFaces-PajorBall-PAMS-1990} which gives a
%convinient volume lower bound for convex bodies that have few faces.
%\begin{theorem}[Ball-Pajor]
%Let $1 \leq p < \infty$ and let $\bm{A} \in \setR^{m \times n}$ with $m \geq n$. Then
%  \[
%  \textrm{vol}_n\Big(\Big\{ \bm{x} \in \setR^n : |\left<\bm{A}_i,\bm{x}\right>| \leq \sqrt{p} \cdot \Big(\frac{1}{n}\sum_{i'=1}^m \|\bm{A}_{i'}\|_2^p\Big)^{1/p} \;\; \forall i \in [m]\Big\}\Big) \geq 1 
%\]
%\end{theorem}
\section{Main technical result}

In this section we show our measure lower bound for balancing vectors from $\ell_p$ to $\ell_q$:

\begin{theorem} \label{thm:MainMeasureBound}
Let $n \le m$ and $1 \le p \le q \le \infty$. Then for any $\bm{a}_1, \dots, \bm{a}_n \in B^m_p$,
\[
\displaystyle \gamma_n\Big(\Big\{
\bm{x} \in \setR^n : \Big\|\sum_{i=1}^n x_i \bm{a}_i\Big\|_q \le \sqrt{\min\Big(p,\log\Big(\frac{2m}{n}\Big)\Big)} \cdot n^{\max(0,1/2 - 1/p) + 1/q}
\Big\}\Big) \ge 2^{-O(n)}.
\]
\end{theorem}

In order to show Theorem~\ref{thm:MainMeasureBound}, roughly speaking it will suffice to show the corresponding bounds for the two special cases of $q \in \{p, \infty\}$, which can be bootstrapped into a general bound.
First we address the simpler case $p = q$ which at heart is based on Khintchine's inequality:

\begin{lemma} \label{lem:Caseqp}
Let $n \le m$ and $p \ge 1$. Then for any $\bm{a}_1, \dots, \bm{a}_n \in B^m_p$,
\[
\displaystyle \gamma_n\Big(\Big\{\bm{x} \in \setR^n: \Big\|\sum_{i=1}^n x_i \bm{a}_i\Big\|_p \le \sqrt{p} \cdot n^{\max(1/2,1/p)}\Big\}\Big) \ge 2^{-O(n)}.
\]
\end{lemma}

\begin{proof}
By Lemma~\ref{lem:ExpectedLpNorm} we know that, for some constant $C > 0$,
\[
\E_{\bm{x} \sim N(\bm{0},\bm{I}_n)}\Big[\Big\|\sum_{i=1}^n x_i \bm{a}_i\Big\|_p\Big] \le C \sqrt{p} \cdot n^{\max(1/2,1/p)}.
\]
By Markov's inequality it follows that 
\[
\displaystyle \gamma_n\Big(\Big\{\bm{x} \in \setR^n: \Big\|\sum_{i=1}^n x_i \bm{a}_i\Big\|_p \le 2C \sqrt{p} \cdot n^{\max(1/2,1/p)}\Big\}\Big) \ge 1/2,
\]
 so that the result follows by Lemma~\ref{lem:ScalingLemma}.
\end{proof}

Next, we deal with the crucial case $q = \infty$:

\begin{lemma} \label{lem:CaseqInfty}
  Let $n \le m$ and $p \ge 1$. Then for any $\bm{A} \in \setR^{m \times n}$ with columns $\bm{a}_1, \dots, \bm{a}_n \in B^m_p$ and rows $\bm{A}_1, \dots, \bm{A}_m \in \setR^n$,
  the body $K := \{ \bm{x} \in \setR^n : \|\sum_{i=1}^n x_i\bm{a}_i\|_{\infty} \leq \sqrt{p} \cdot n^{\max(0,1/2 - 1/p)}\}$
  satisfies
\[
\gamma_n(K) \ge \prod_{j \in [m]} \gamma_n (\{ \bm{x} \in \setR^n : |\inn{\bm{x}}{\bm{A}_j}| \le \sqrt{p} n^{\max(0,1/2 - 1/p)} \}) \ge 2^{-O(n)}.
\]
\end{lemma}
\begin{proof}
%Geometrically speaking the idea is to lower bound $\gamma_n(K)$ by using the Lemma Sidak-Kathri
  The main idea in the proof is that we can convert the bound on the $\ell_p$-norm of the columns $\bm{a}_i$ into information about the $\ell_2$-norm of the rows $\bm{A}_j$. Namely, 
\begin{equation} \label{eq:BoundOnL2NormOfRowVecs}
\Big(\frac{1}{n} \sum_{j \in [m]} \|\bm{A}_j\|_2^p\Big)^{1/p} \stackrel{\textrm{Lem~\ref{lem:ElementaryInequalityOnLpNorms}}}{\le} n^{\max(0,1/2-1/p) } \cdot \Big(\frac{1}{n} \underbrace{\sum_{j \in [m]} \|\bm{A}_j\|_p^p}_{\le n}\Big)^{1/p} \le n^{\max(0,1/2-1/p) }.
\end{equation}
We rescale the row vectors to $\bm{V}_j := (\sqrt{p} n^{\max(0,1/2-1/p)})^{-1} \bm{A}_j$ and abbreviate $y_j := \|\bm{V}_j\|_2^2$, so that Eq.~\eqref{eq:BoundOnL2NormOfRowVecs} simplifies to  $\sum_{j=1}^m y_j^{p/2} \le n \cdot p^{-p/2}$. We may then apply \v{S}idak's Lemma~\ref{lem:SidakLemma} and bound the one-dimensional measure:
\begin{eqnarray*}
  \gamma_n (K)  &=& \gamma_n \big(\big\{ \bm{x} \in \setR^n : |\left<\bm{x},\bm{V}_j\right>| \leq 1 \; \; \forall j \in [m]\big\}\big) \\ & \stackrel{\textrm{Lem~\ref{lem:SidakLemma}}}{\ge}& \prod_{j \in [m]} \gamma_n \big(\big\{\bm{x} \in \setR^n : |\inn{\bm{x}}{\bm{V}_j}| \le 1 \big\}\big) \\
                &\stackrel{\textrm{Lem~\ref{lem:OneDimEstimate}}}{\ge}& \prod_{j \in [m]} \big(1-\exp(-y_j^{-1}/2)\big) \\
  &\stackrel{\textrm{Claim I}}{\geq}& \prod_{j \in [m]} \exp\Big(-C' p^{p/2}y_j^{p/2}\Big) = \exp\Big(-C'p^{p/2}\sum_{j \in [m]} y_j^{p/2}\Big) \geq \exp(-C'n) 
\end{eqnarray*}
Here we have used an estimate that remains to be proven: \\
{\bf Claim I.} \emph{For any $p \geq 1$ and $y>0$ one has $1-\exp(-\frac{1}{2y}) \geq \exp(-C' p^{p/2}y^{p/2})$ where $C'>0$ is a universal constant.} \\
{\bf Proof of Claim I.} 
It will suffice to show for any $y > 0$:
\[
-\log(1-\exp(-y^{-1}/2)) \le O(p^{p/2} y^{p/2}).
\]
% Indeed, if this holds then adding all $m$ inequalities and exponentiating both sides gives the bound.
To see this, let $z = \sqrt{2y}$ and note that it suffices to show 
\[
-\log(1-\exp(-z^{-2})) \cdot z^{-p} \le O((p/2)^{p/2}).
\]

First, by convexity of $x \mapsto -\log(1-x)$, we have $-\log(1-x) \le O(x)$ for $x \in [0, 1/e]$. It follows that for $z \le 1$, we have
\[
-\log(1-\exp(-z^{-2})) \le O(\exp(-z^{-2})) \le O(\lceil p/2 \rceil !/z^{-2\lceil p/2 \rceil}),
\]

and therefore $-\log(1-\exp(-z^{-2})) \cdot z^{-p} \le O(\lceil p/2 \rceil !) \le O((p/2)^{p/2})$.

Next, we claim that $-\log(1-\exp(-z^{-2})) \le 4z$ for all $z > 0$. Indeed, both sides tend to $0$ as $z \to 0$ and the derivative of the left side is
\[
\frac{2}{z^3 \Big(\exp\Big(\frac{1}{2z^2}\Big) - 1\Big)} < \frac{2}{z^3 \Big(\frac{1}{2z^2} + \frac{1}{8z^4}\Big)} = \frac{16z}{4z^2+1} \le 4,
\] where we used $e^x > 1 + x + x^2/2$ for $x = \frac{1}{2z^2}$ and $(2z-1)^2 \ge 0$. It follows that when $z \ge 1$, $-\log(1-\exp(-z^{-2})) \cdot z^{-p} \le 4z^{1-p} \le 4 \le O((p/2)^{p/2})$.
\end{proof}

\begin{remark}
  This argument is largely motivated by the result of Ball and Pajor~\cite{ConvexBodiesWithFewFaces-PajorBall-PAMS-1990} which bounds volume instead of Gaussian measure. More specifically, \cite{ConvexBodiesWithFewFaces-PajorBall-PAMS-1990} prove that for $1 \leq p \leq \infty$ and any matrix $\bm{A} \in \setR^{m \times n}$, the set
  \[
    K = \Big\{ \bm{x} \in \setR^n : |\left<\bm{A}_j,\bm{x}\right>| \leq \sqrt{p} \cdot \Big(\frac{1}{n} \sum_{j=1}^m \|\bm{A}_j\|_2^p\Big)^{1/p} \; \forall j \in [m]  \Big\}
  \]
  satisfies $\textrm{vol}_n(K) \geq 1$. In contrast, our Lemma~\ref{lem:CaseqInfty} provides a simpler proof of a stronger result (up to a constant scaling), since the volume of a convex body is always at least its Gaussian measure. On the other hand, it is also possible to recover Lemma~\ref{lem:CaseqInfty} directly from this result together with Theorem~\ref{thm:HereditaryVolumeImpliesGaussianMeasure}.
\end{remark}

We are now ready to show Theorem~\ref{thm:MainMeasureBound}:

\begin{proof}[Proof of Theorem~\ref{thm:MainMeasureBound}]
Let $1 \leq p \leq q \leq \infty$ and 
let $\bm{A} \in \setR^{m \times n}$ denote the matrix with columns $\bm{a}_1, \dots, \bm{a}_n \in B_p^m$. By 
Lemma~\ref{lem:InterpolationInequality} we know that for any $\bm{z} \in \setR^m$ with $\|\bm{z}\|_p \leq n^{1/p}$ and $\|\bm{z}\|_{\infty} \leq 1$
one has $\|\bm{z}\|_q \leq ( \|\bm{z}\|_p^p \cdot \|\bm{z}\|_{\infty}^{q-p})^{1/q} \leq n^{1/q}$. Phrased in geometric terms this means
$n^{1/q} B^m_q \supseteq n^{1/p} B^m_p \cap B^m_\infty$. We would like to point out that this is a crucial point to obtain a dependence solely on $n$ rather than the larger parameter $m$.
Next, note the fact that  $\bm{A}^{-1}(S \cap T) = \bm{A}^{-1} (S) \cap \bm{A}^{-1} (T)$ for any sets $S$ and $T$
which we use together with the inequality of \v{S}idak and Kathri (Lemma~\ref{lem:SidakLemma}) to obtain the estimate
\begin{eqnarray*}
 & &  \gamma_n\Big(\bm{A}^{-1} \big(\sqrt{p} \cdot n^{\max(0,1/2 - 1/p) + 1/q} B^m_q\big)\Big) \\
  & \ge& \gamma_n\Big(\bm{A}^{-1} \big(\sqrt{p} \cdot n^{\max(0,1/2 - 1/p)} (n^{1/p} B^m_p \cap B^m_\infty)\big)\Big) \\
&\ge& \gamma_n\Big(\bm{A}^{-1} \big(\sqrt{p}\cdot n^{\max(1/2,1/p)} B^m_p\big)\Big) \cdot \prod_{j \in [m]} \gamma_n \big(\big\{ \bm{x} \in \setR^n : |\inn{\bm{x}}{\bm{A}_j}| \le \sqrt{p} n^{\max(0,1/2 - 1/p)} \big\}\big) \\ & \ge& 2^{-O(n)} \cdot 2^{-O(n)} = 2^{-O(n)},
\end{eqnarray*}
where we have used the measure lower bounds from Lemmas~\ref{lem:Caseqp} and~\ref{lem:CaseqInfty}.
This shows the claimed bound whenever $p \leq O(\log(\frac{2m}{n}))$, where the hidden constant can be removed by scaling the corresponding convex body, see Lemma~\ref{lem:ScalingLemma}. 

It remains to prove that we can bootstrap the existing bound for the regime of large $p$. So let us assume that  $p \ge 2\cdot \max\{1, \log(m/n)\}$.
Let $p_0 \in [2,p]$ be a parameter to be determined and remark that Lemma~\ref{lem:ElementaryInequalityOnLpNorms} gives $\|\bm{a}_i\|_{p_0} \le m^{1/p_0 - 1/p} \cdot \|\bm{a}_i\|_p \le m^{1/p_0 - 1/p}$. Applying the above measure lower bound for $p_0$ implies

\[
\displaystyle \gamma_n\Big(\Big\{\bm{x} \in \setR^n : \Big\|\sum_{i=1}^n x_i \bm{a}_i\Big\|_q \le \sqrt{p_0} \cdot n^{1/2 - 1/p_0 + 1/q} \cdot  m^{1/p_0 - 1/p}\Big\}\Big) \ge 2^{-O(n)}.
\]

We can rewrite the above upper bound on $\ell_q$-norm as

\[
\sqrt{p_0} \cdot n^{1/2 - 1/p_0 + 1/q} \cdot  m^{1/p_0 - 1/p} = n^{1/2 - 1/p + 1/q} \cdot \underbrace{\Big(\frac{m}{n}\Big)^{-1/p}}_{\le 1} \cdot \sqrt{p_0} \cdot \Big(\frac{m}{n}\Big)^{1/p_0}.
\]

Taking $p_0 := 2\cdot \max\{1, \log(m/n)\}$ gives the desired result as then $(m/n)^{1/p_0} \le \sqrt{e}$ and Lemma~\ref{lem:ScalingLemma} can again deal with such constant scaling.
\end{proof}

Now our main result on existence of partial colorings easily follows:  
\begin{proof} [Proof of Theorem~\ref{thm:PartialColoringForLpLq}] Apply Theorem~\ref{thm:PartialColoringForExpSmallSetsWithShiftAndLargeFractionColored}
to the set 
\[
K := \Big\{ \bm{x} \in \setR^n : \Big\|\sum_{i=1}^n x_i\bm{a}_i\Big\|_q \leq \sqrt{\min\Big(p, \log\Big(\frac{2m}{n}\Big)\Big)} \cdot n^{1/2-1/p+1/q}\Big\},
\]
 which by Theorem~\ref{thm:MainMeasureBound} indeed has a Gaussian measure of $\gamma_n(K) \geq 2^{-O(n)}$.
\end{proof}
Next, we show how to obtain a full coloring by iteratively finding partial colorings.
\begin{proof} [Proof of Theorem~\ref{thm:FullColoringForLpLq}] % $
Let again $2 \leq p \leq q \leq \infty$ and let $\bm{a}_1,\ldots,\bm{a}_n \in B_p^m$. 
 We begin with $\bm{x}^{(0)} := \bm{0}$ and given $\bm{x}^{(0)},\ldots,\bm{x}^{(t)}$ we set $S^{(t)} := \{ i \in [n] : -1 < x_i^{(t)} < 1\}$ as the \emph{active
variables}. Then combining Theorem~\ref{thm:PartialColoringForExpSmallSetsWithShiftAndLargeFractionColored} and Theorem~\ref{thm:MainMeasureBound} we can find a partial coloring $\bm{x}^{(t+1)} \in [-1,1]^n$ in polynomial time so that $|S^{(t+1)}| \leq |S^{(t)}| / 2$ and $\|\sum_{i=1}^n (x_i^{(t+1)} - x_i^{(t)})\bm{a}_i\|_q \leq C_1\sqrt{\min(p, \log (\frac{2m}{|S^{(t)}|}))} \cdot |S^{(t)}|^{1/2-1/p + 1/q}$. Let $\bm{x}^{(T)}$ be the first iterate with $\bm{x}^{(T)} \in \{ -1,1\}^n$. Clearly $|S^{(t)}| \leq n2^{-t}$ and $T \leq \log_2(n)$. 
Using the triangle inequality we get
\begin{eqnarray*}
 \Big\|\sum_{i=1}^n x^{(T)}_i\bm{a}_i\Big\|_q &\leq& \sum_{t=0}^{T-1} \Big\|\sum_{i=1}^n (x^{(t+1)}_i-x^{(t)}_i)\bm{a}_i\Big\|_q \\
                                 &\leq& C_1\sum_{t = 0}^{T-1} \sqrt{\min\Big(p, \log\Big(\frac{2m}{2^{-t} \cdot n} \Big)\Big) } \cdot (2^{-t} \cdot n)^{1/2-1/p+1/q}  \\
  &\le& \frac{C_1C_2 \sqrt{\min\Big(p,\log\Big(\frac{2m}{n}\Big)\Big)}}{1/2 - 1/p + 1/q}   \cdot n^{1/2 - 1/p + 1/q}. \qedhere
\end{eqnarray*}
\end{proof}
The intuition behind the extra factor for obtaining a full coloring is as follows: abbreviate the exponent as
$\beta := 1/2-1/p+1/q$. Then it takes $\frac{1}{\beta}$ iterations until the term $|S^{(t)}|^{\beta}$ decreases by a factor of 1/2 which dominates the miniscule growth of the logarithmic term. Then indeed the overall
discrepancy is dominated by the discrepancy from the first $\frac{1}{\beta}$ iterations.

We can now demonstrate how a nontrivial choice of $\ell_p$-norms can be beneficial in classical
discrepancy settings: 
\begin{proof}[Proof of Corollary~\ref{cor:BeckFiala}]
  Consider column vectors $\bm{a}_1,\ldots,\bm{a}_n \in \{ 0,1\}^m$ with at most $t$ nonzero entries per $\bm{a}_i$.
First let us study the case $t \ge n/10$. Since for each column $\|\bm{a}_i\|_4 \le t^{1/4}$, Theorem~\ref{thm:FullColoringForLpLq} provides a coloring $\bm{x} \in \{-1,1\}^n$ with $\|\sum_{i=1}^n x_i \bm{a}_i\|_\infty \le O(n^{1/4} \cdot t^{1/4}) = O(\sqrt{t})$. \footnote{In fact for $t \ge n$ a more careful choice of $p = \log(2t/n)$ gives a better $\ell_\infty$ discrepancy bound of $O(\sqrt{n \log(2t/n)})$, even though the Beck-Fiala conjecture asks only for $O(\sqrt{t})$.}

Now if $t < n/10$, we take $p \in [2,16)$ with $1/2 - 1/p = 1/\log(n/t)$. Then $\|\bm{a}_i\|_p \le t^{1/p}$ and Theorem~\ref{thm:FullColoringForLpLq} gives $\bm{x} \in \{-1,1\}^n$ with 
\[
\displaystyle \Big\|\sum_{i=1}^n x_i \bm{a}_i\Big\|_\infty \le \frac{C \cdot n^{1/2-1/p} \cdot t^{1/p}}{1/2 - 1/p} = C\sqrt{t} \log(n/t) \cdot \underbrace{(n/t)^{1/\log(n/t)}}_{= e}. \qedhere
\]
\end{proof}
We conclude this section by showing that the term $n^{\max(0,1/2 - 1/p) + 1/q}$ in our bounds is necessary:
\begin{proof} [Proof of Theorem~\ref{thm:LowerBound}]
Consider the case $p \ge 2$. Consider an $n \times n$ \emph{Hadamard matrix}, which is a matrix  $\bm{H} \in \{-1,1\}^{n \times n}$ so that all rows and columns are orthogonal. Such matrices are known to exist at least whenever $n$ is a power of 2. The columns satisfy $\|\bm{h}_i\|_p = n^{1/p}$ and for any $\bm{x} \in [-1,1]^n$ with  $|\{i : x_i^2 = 1\}| \ge n/2$  we know that $\|\bm{x}\|_2 \ge \Omega(\sqrt{n})$ and $\|\bm{Hx}\|_2 \ge \Omega(n)$, so that by Lemma~\ref{lem:ElementaryInequalityOnLpNorms} we have
\[
\|\bm{Hx}\|_q \ge \|\bm{Hx}\|_2 \cdot n^{1/q - 1/2} = \Omega(n^{1/2+1/q}).
\]
For $p \in [1,2]$, take an identity matrix $\bm{I}_n$. For every $\bm{x} \in [-1,1]^n$ with $|\{i : x_i^2 = 1\}| \ge n/2$ we have $\|\bm{I}_n \bm{x}\|_q = \|\bm{x}\|_q \ge \Omega(n^{1/q})$, and the columns of $\bm{I}_n$ are certainly in $B^m_p$.
\end{proof}

\section{Partial coloring via measure lower bound}

%TODO: copy source file from "Chapter 19: Partial Colorings for exponentially small sets"

In this section, we want to show the existence of partial fractional colorings for bodies $K$ with
$\gamma_n(K) \geq e^{-\alpha n}$ as promised in Theorem~\ref{thm:PartialColoringForExpSmallSetsWithShiftAndLargeFractionColored}. The main innovation of this work compared to e.g. \cite{ConstructiveDiscrepancy-Rothvoss-FOCS2014}
is to handle an arbitrarily small constant $\alpha>0$. We will show how to find a partial coloring that
colors a small constant fraction of coordinates; then iterating the argument will color the promised $\beta-\gamma$ fraction.
Also, instead of working with a shift $\bm{y}$ and a scaling of $K$, it will be notationally easier to work with a shifted and scaled box.
Hence, for vectors $\bm{L},\bm{R} \in \setR_{\geq 0}^n$, we write $[-\bm{L},\bm{R}] := [-L_1,R_1] \times \ldots \times [-L_n,R_n]$ as the box defined by constraints $-L_i \leq x_i \leq R_i$ for $i=1,\ldots,n$.
We use $N(\bm{0},H)$ to denote the Gaussian distribution restricted to a subspace $H \subseteq \setR^n$. %for the standard Gaussian such a subspace $H$.
Then the main technical result for this section will be:
\begin{theorem} \label{thm:PartialColoringForExpSmallSets}
  For all constants $\alpha,\beta > 0$ there are  $\varepsilon := \varepsilon(\alpha,\beta) > 0$ and $\delta := \delta(\alpha,\beta)>0$ so that the following holds: 
  Let  $K \subseteq \setR^n$ be a symmetric convex body
  with $K \subseteq H$ for a subspace $H \subseteq \setR^n$ with $\dim(H) \geq \beta n$ and
   $\gamma_H(K) \geq e^{-\alpha n}$; also let $\bm{L},\bm{R} \in [0,\varepsilon]^n$. 
  Assuming a weak separation oracle for $K$, there is a randomized polynomial time algorithm which finds an
  $\bm{x} \in K \cap [-\bm{L},\bm{R}]$ so that $|\{ i \in [n] : x_i \in \{ -L_i,R_i\}\}| \geq \delta n$ with probability at least $1 - e^{-\Theta_{\varepsilon,\delta}(n)}$.
\end{theorem}
Note that the considered box satisfies $[-\bm{L},\bm{R}] \subseteq [-\varepsilon,\varepsilon]^n$.
We would like to point out that applying the standard nonconstructive proof by Gluskin~\cite{RootNDisc-Gluskin89} and Giannopoulos~\cite{Giannopoulos1997} to a find a partial coloring  $\bm{x} \in \{ -\varepsilon,0,\varepsilon\}^n$ with support $\Omega(n)$ will
require either a \emph{small enough} constant $\alpha>0$, or $\varepsilon$ needs to be exponentially small in $n$.
%In fact, the statement of Theorem~\ref{thm:PartialColoringForExpSmallSets} does not hold if $\bm{x} \in [-varepsilon,\varepsilon]^n$ is replaced by $\bm{x} \in \{ -\varepsilon,0,\varepsilon\}^n$. 
In fact, it is not hard to construct a thin strip $K$ with $\gamma_n(K) \geq e^{-\Omega(n)}$ so that $K$ does not intersect $\{ -1,0,1\}^n \setminus \{ \bm{0}\}$ (even after a subexponential scaling). We show the construction in Appendix~\ref{appendix:SetsWithoutPartialColorings}.
%In fact, it is not hard to con
%Moreover, for a choice of say $K := \{ \bm{x} \in \setR^n \mid d(\bm{x},H) \leq 2^{-10n}\}$ with $H \subseteq \setR^n$ being an $(n-1)$-dimensional hyperplane one might even expect that $K$ does not contain any point in $\{ -1,0,1\}^n \setminus \{ \bm{0} \}$. We want to prove that nevertheless a good \emph{fractional} partial coloring exists. 

%We will assume that $\alpha\geq 1$ and $n \geq n_0(\alpha,\varepsilon)$ is large enough.
For our proof we make use of the  mean width $w(Q) := \E_{\bm{\theta} \in S^{n-1}}[\sup_{\bm{x} \in Q} \left<\bm{\theta},\bm{x}\right>]$ of a body. We should point out that
the connection between partial coloring arguments and mean width is due to Eldan and Singh~\cite{DBLP:journals/rsa/EldanS18}. Several of the claims require that $n$ is chosen large enough.
%Moreover $N(0,\bm{I}_n)$ is the standard Gaussian with density function $\gamma_n(\bm{x}) := \frac{1}{(2\pi)^{n/2}}e^{-\|\bm{x}\|_2^2/2}$.
\begin{lemma}
  Let $Q \subseteq \setR^n$ be a symmetric convex body with $\gamma_n(Q) \geq e^{-\alpha n}$ for $\alpha > 0$. Then
  $w(Q) \geq \frac{1}{2} e^{-\alpha}\sqrt{n}$.
\end{lemma}
\begin{proof}
  Let $r>0$ be the radius so that $\gamma_n(rB_2^n) = \gamma_n(Q)$. 
  By \emph{Urysohn's Inequality} (Theorem~\ref{thm:UrysohnInequality}) one has $w(Q) \geq w(rB_2^n) = r$ so it suffices to give a lower bound on the radius $r$. 
  A simple but useful estimate is that  $2^n \leq \textrm{Vol}_n( \sqrt{n} B_2^n) \leq 5^n$ for any $n \geq 1$.
  Moreover, the Gaussian density is maximized at $\gamma_n(\bm{0}) = \frac{1}{(\sqrt{2\pi})^n}$. 
  Then for $\beta := 2e^{\alpha} \geq 2$ we have
  \[
    \gamma_n\Big( \frac{\sqrt{n}}{\beta} B_2^n\Big) \leq \textrm{Vol}_n\Big( \frac{\sqrt{n}}{\beta} B_2^n\Big) \cdot \gamma_n(\bm{0})
    \leq \Big(\frac{5}{\beta}\Big)^n \cdot \frac{1}{(\sqrt{2\pi})^{n}} \leq \Big(\frac{2}{\beta}\Big)^n \stackrel{\beta = 2e^{\alpha}}{\leq} e^{-\alpha n}
  \]
  and so $r \geq \frac{\sqrt{n}}{\beta} = \frac{\sqrt{n}}{2e^{\alpha}}$.
  \end{proof}

%Note that this claim $$
The key modification of our work in contrast to \cite{ConstructiveDiscrepancy-Rothvoss-FOCS2014} is a finer upper bound on the distance of a Gaussian to $K$:
  \begin{lemma} \label{lem:DistanceGaussianToExpSmallSet}
  Let $K \subseteq \setR^n$ be a symmetric convex set with $\gamma_n(K) \geq e^{-\alpha n}$ where $\alpha \geq 1$ and $n$ is large enough. Then
  \[
  \E_{\bm{x} \sim N(\bm{0},\bm{I}_n)}[d(\bm{x},K)] \leq \sqrt{n} \cdot \Big(1-\frac{1}{512\alpha e^{4\alpha}}\Big)
  \]
\end{lemma}
\begin{proof}
  Note that by Theorem ~\ref{thm:ConcentrationFor1LipschitzFunctions} we have $\Pr_{\bm{x} \sim N(\bm{0},\bm{I}_n)}[\|\bm{x}\|_2 \geq 4\sqrt{\alpha n}] \leq e^{-2\alpha n}$,
  hence the restriction
  $Q := K \cap 4\sqrt{\alpha n} B_2^n$ still has $\gamma_n(Q) \geq \gamma_n(K) - e^{-2\alpha n} \geq e^{-2\alpha n}$ for $n$ large enough.
  Then by the previous Lemma we know that $w(Q) \geq \frac{\sqrt{n}}{2e^{2\alpha}}$. For a vector $\bm{x}$, let $\bm{z}(\bm{x}) := \textrm{argmax}\{ \left<\bm{z},\bm{x}\right> : \bm{z} \in Q\}$. As we just showed, $\E_{\bm{x} \sim N(\bm{0},\bm{I}_n)}[ \left<\bm{z}(\bm{x}),\frac{\bm{x}}{\|\bm{x}\|_2}\right>] \geq \frac{\sqrt{n}}{2e^{2\alpha}}$. Let $\lambda \in [0,1]$ be a parameter that we determine later. Note that the point $\lambda \cdot \bm{z}(\bm{x})$ lies in $Q$.
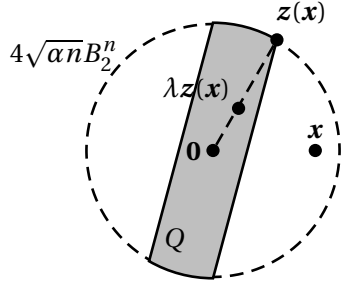
  \begin{center}
    \psset{unit=1.7cm}
    \begin{pspicture}(-1.3,-1.1)(1.3,1.1)
      \SpecialCoor
    \pscircle[linestyle=dashed](0,0){1}
    \pscustom[fillstyle=solid,fillcolor=lightgray]{
      \psarc(0,0){1}{60}{90}
      \psarc(0,0){1}{-120}{-90}
      \psarc(0,0){1}{60}{61}      
    }%
    \rput[c](-0.3,-0.7){$Q$}
    \rput[r](-0.75,0.75){$4\sqrt{\alpha n}B_2^n$}
    \cnode*(0,0){2.5pt}{origin} \nput[labelsep=2pt]{180}{origin}{$\bm{0}$}%
    \cnode*(0.8,0){2.5pt}{x} \nput[labelsep=2pt]{90}{x}{$\bm{x}$}
    \cnode*(1;60){2.5pt}{z} \nput{60}{z}{$\bm{z}(\bm{x})$}
    \ncline[linestyle=dashed]{origin}{z}
    \cnode*(0.2,0.33){2.5pt}{lambdaz} \nput[labelsep=2pt]{135}{lambdaz}{$\lambda \bm{z}(\bm{x})$}
      \end{pspicture}
    \end{center}
  This point can be used to bound
  \begin{eqnarray*}
   \E_{\bm{x} \sim N(\bm{0},\bm{I}_n)}[\|\bm{x} - \lambda \bm{z}(\bm{x})\|_2^2] &=& \E[\|\bm{x}\|_2^2] - 2\lambda \E[\left<\bm{x},\bm{z}\right>] + \E[\lambda^2\|\bm{z}\|_2^2] \\
                 &=& \underbrace{\E[\|\bm{x}\|_2^2]}_{=n} - 2\lambda \underbrace{\E[\|\bm{x}\|_2]}_{\geq\frac{1}{2} \sqrt{n}} \cdot \underbrace{\E_{\bm{\theta} \in S^{n-1}}[\left<\bm{\theta},\bm{z}(\bm{\theta})\right>]}_{\geq \sqrt{n}/(2e^{2\alpha})} + \E[\lambda^2\underbrace{\|\bm{z}\|_2^2}_{\leq 16\alpha n}] \\
    &\leq& n - \frac{1}{2}e^{-2\alpha} \lambda n  + \lambda^2 \cdot 16 \alpha n \stackrel{\lambda := \frac{1}{64\alpha e^{2\alpha}}}{=} n \cdot \Big(1-\frac{1}{256\alpha e^{4\alpha}}\Big)
  \end{eqnarray*}
  Then
  \[
    \E[d(\bm{x},Q)] \stackrel{\lambda \bm{z} \in Q}{\leq} \E[\|\bm{x}-\lambda \bm{z}\|_2] \stackrel{\textrm{Jensen}}{\leq} \E[\|\bm{x}-\lambda\bm{z}\|_2^2]^{1/2} \leq \sqrt{n} \cdot \sqrt{1-\frac{1}{256 \alpha e^{4\alpha}}} \leq \sqrt{n} \cdot \Big(1-\frac{1}{512\alpha e^{4\alpha}}\Big)
  \]
  using $\sqrt{1-y} \leq 1-\frac{y}{2}$ for $0 \leq y \leq 1$.
\end{proof}

Lemma~\ref{lem:DistanceGaussianToExpSmallSet} can be extended to the case that $K$ is included in a not too small subspace $H$. 
\begin{lemma} \label{lem:DistanceGaussianToExpSmallSetDim}
  Let $\alpha \geq 1$, $0 < \beta \leq 1$ be constants.
  Let $H \subseteq \setR^n$ be a subspace with $\dim(H) \geq \beta n$ 
  and let  $K \subseteq H$ be a symmetric convex body with 
   $\gamma_H(K) \geq e^{-\alpha n}$. For $n$  large enough, one has
  \[
  \E_{\bm{x} \sim N(\bm{0},\bm{I}_n)}[d(\bm{x},K)] \leq \sqrt{n} \cdot \Big(1-\frac{\beta}{512\alpha e^{4\alpha}}\Big)
  \]
\end{lemma}
\begin{proof}
  Note that one can generate a Gaussian $\bm{x} \sim N(\bm{0},\bm{I}_n)$ as $\bm{x} = \bm{x}_1 + \bm{x}_2$ where $\bm{x}_1 \sim N(\bm{0},H^{\perp})$ and $\bm{x}_2 \sim N(\bm{0},H)$ independently. Then $d(\bm{x},K)^2 = d(\bm{x}_1,H)^2 + d(\bm{x}_2,K)^2$ by Pythagoras.
  Hence
  \begin{eqnarray*}
    \E_{\bm{x} \sim N(\bm{0},\bm{I}_n)}\big[d(\bm{x},K)^2\big] &\leq& \E_{\bm{x}_1 \sim N(\bm{0},H^{\perp})}\big[d(\bm{x}_1,H)^2\big] + \E_{\bm{x}_2 \sim N(\bm{0},H)}\big[d(\bm{x}_2,K)^2\big] \\ &\stackrel{\textrm{Lem~\ref{lem:DistanceGaussianToExpSmallSet}}}{\leq}& \dim(H^{\perp}) + \dim(H) \cdot \Big(1-\frac{1}{256 \alpha e^{4\alpha}}\Big)     \\                                                                                     &\stackrel{\dim(H) \geq \beta n}{\leq}& n \cdot \Big(1 - \frac{\beta}{256 \alpha e^{4\alpha}}\Big) % %\leq (1-\beta) n + \beta n (1-\frac{1}{256 \alpha e^{4\alpha}}) \\ &\leq&
  \end{eqnarray*}
  As in the proof of Lemma~\ref{lem:DistanceGaussianToExpSmallSet}, the claim follows after applying Jensen inequality
  with the fact that $\sqrt{1-y} \leq 1-\frac{y}{2}$ for $0 \leq y \leq 1$.
\end{proof}

Next, we show the average distance of a Gaussian to the cube $[-\varepsilon,\varepsilon]^n$ is $\sqrt{n} \cdot (1-\Theta(\varepsilon))$.
\begin{lemma} \label{lem:DistanceOfGaussianFromEpsilonCube}
Let $\varepsilon>0$. Then for $n$ large enough one has
  \[
 \Pr_{\bm{x} \sim N(\bm{0},\bm{I}_n)}\Big[d(\bm{x},[-\varepsilon,\varepsilon]^n) \geq (1-5\varepsilon) \sqrt{n}\Big] \geq 1 - \exp\Big(-\frac{\varepsilon^2}{2}  n\Big)
  \]  
\end{lemma}
\begin{proof}
  Let $\bm{y} := \bm{y}(\bm{x}) := \textrm{argmin}\{ \|\bm{x} - \bm{y}\|_2 : \bm{y} \in [-\varepsilon,\varepsilon]^n \}$ be the closest point in the cube to $\bm{x}$.
  For an individual coordinate $i \in [n]$ the expected contribution to the distance is 
  \[
    \E\big[d(x_i,[-\varepsilon,\varepsilon])^2\big] = \E\big[|x_i - y_i|^2\big] = \underbrace{\E[x_i^2]}_{=1} - 2\underbrace{\E[x_iy_i]}_{\leq \varepsilon \E[|x_i|]}+\underbrace{\E[y_i^2]}_{\geq 0}
 \geq   1 - 2 \sqrt{\frac{2}{\pi}} \cdot \varepsilon \geq 1-2\varepsilon.
  \]
  Then by linearity $\E[d(\bm{x},[-\varepsilon,\varepsilon]^n)^2]^{1/2} \geq \sqrt{n \cdot (1-2\varepsilon)} \geq \sqrt{n} \cdot (1-2\varepsilon)$. Recall that the distance function  $F(\bm{x}) := d(\bm{x},[-\varepsilon,\varepsilon]^n)$ is 1-Lipschitz and for such functions the difference $|\E[F(\bm{x})]-\E[F(\bm{x})^{2}]^{1/2}|$ is bounded by an absolute constant. Then $\E[F(\bm{x})] \geq \sqrt{n} \cdot (1-4\varepsilon)$ for $n$ large enough.
 Finally by Theorem~\ref{thm:ConcentrationFor1LipschitzFunctions} one has $\Pr[F(\bm{x}) < \E[F(\bm{x})] - \varepsilon \sqrt{n}] \leq e^{-\varepsilon^2 n / 2}$ for $\bm{x} \sim N(\bm{0},\bm{I}_n)$ which then gives the claim as  $\E[F(\bm{x})] - \varepsilon \sqrt{n} \geq (1-5\varepsilon) \sqrt{n}$. %Note that for 1-Lipschitz functions, the difference $|\E[F(\bm{x})]-\E[F(\bm{x})^{2}]^{1/2}|$ is bounded by an absolute constant.
%  Finally $\E[F(\bm{x})] + \frac{\varepsilon}{\sqrt{50}} \leq \E[F(\bm{x})^2]^{1/2} + O(1) + \frac{\varepsilon}{\sqrt{50}} \leq \sqrt{n} \cdot (1-\frac{\varepsilon}{16})$ for $n$ large enough. 
%  By adding in some slack and using concentration (recall that the function $\bm{x} \mapsto d(\bm{x},[-\varepsilon,\varepsilon]^n)$ is 1-Lipschitz), the claim follows. 
  %
%$\E[d(x_i,[-\varepsilon,\varepsilon])^2] \geq \E_{x_i \sim N(0,1)}[ x_i-\varepsilon \mid x_i \geq 0] = \underbrace{\E_{x_i \sim N(0,1)}[|x_i|] - \varepsilon$
\end{proof}

We will now prove Theorem~\ref{thm:PartialColoringForExpSmallSets}. Let $H \subseteq \setR^n$ be a subspace with $\dim(H) \geq \beta n$ and let $K \subseteq H \subseteq \setR^n$ be a
symmetric convex body with $\gamma_H(K) \geq e^{-\alpha n}$. Moreover, let $L_i,R_i \in [0,\varepsilon]$ be given parameters where the choice of $\varepsilon := \varepsilon(\alpha,\beta) > 0$ will be made in the upcoming proof of Lemma~\ref{lem:PartialColoringForExpSmallSetsMainLemma}.
%Instead of providing a vector $\bm{x} \in (\frac{1}{\varepsilon} K) \cap [-1,1]^n$ directly, we will instead
%find an $\bm{x} \in K \cap [-\varepsilon,\varepsilon]^n$ with $|\{i \in [n] : x_i \in \{ -\varepsilon,\varepsilon\}\}| \geq \delta n$ where $\varepsilon,\delta>0$ will be chosen small enough, depending on $\alpha$ --- the result in Theorem~\ref{thm:PartialColoringForExpSmallSets} then follows by scaling $\bm{x}$ by $\frac{1}{\varepsilon}$.  
We will use the following algorithm:
\begin{enumerate}
\item[(1)] Pick $\bm{x}^* \sim N(\bm{0},\bm{I}_n)$ at random.
\item[(2)] Compute $\bm{y}^* := \textrm{argmin}\big\{ \|\bm{x}^* - \bm{y}\|_2 : \bm{y} \in K \cap [-\bm{L},\bm{R}]^n \big\}$.
\end{enumerate}
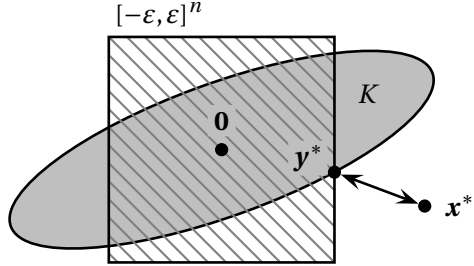
\begin{center}
\psset{unit=1.5cm}
\begin{pspicture}(-2.4,-0.9)(1.5,1.2)
\rput[c]{20}(0,0){\psellipse[fillstyle=solid,fillcolor=lightgray](0,0)(2,0.6)}
%\rput[c]{20}(0,0){\psellipticwedge[fillstyle=solid,fillcolor=lightgray](0,0)(2,0.6){0}{-0.1}}
\pspolygon[fillstyle=vlines,fillcolor=lightgray,hatchcolor=gray](-0.8,-1)(1,-1)(1,1)(-0.8,1)
\rput[c](1.3,0.5){\psframebox[framesep=2pt,fillstyle=solid,fillcolor=lightgray,linestyle=none]{$K$}}
\cnode*(0,0){2.5pt}{origin} \nput[labelsep=2pt]{90}{origin}{\psframebox[fillstyle=solid,fillcolor=lightgray,framesep=2pt,linestyle=none]{$\bm{0}$}}
\cnode*(1.8,-0.5){2.5pt}{x} \nput{0}{x}{$\bm{x}^*$}
\cnode*(1,-0.2){2.5pt}{y} \nput[labelsep=0pt]{150}{y}{\psframebox[fillstyle=solid,fillcolor=lightgray,framesep=1pt,linestyle=none]{$\bm{y}^*$}}
\rput[l](-0.95,1.2){$[-\bm{L},\bm{R}]$} % \subseteq [-\varepsilon,\varepsilon]^n$}
\ncline[arrowsize=6pt,linewidth=1pt]{<->}{x}{y}
\end{pspicture}
%\begin{pspicture}(-2.4,-0.9)(1.5,1.2)
%\drawRect{fillstyle=solid,fillcolor=lightgray}{-1}{-1}{2}{2}
%\rput[c]{20}(0,0){\psellipse[fillstyle=vlines,hatchcolor=darkgray](0,0)(2,0.6)}
%\rput[c](1.3,0.5){\psframebox[framesep=2pt,fillstyle=solid,fillcolor=white,linestyle=none]{$K$}}
%\cnode*(0,0){2.5pt}{origin} \nput[labelsep=2pt]{90}{origin}{\psframebox[fillstyle=solid,fillcolor=lightgray,framesep=2pt,linestyle=none]{$\bm{0}$}}
%\cnode*(1.8,-0.5){2.5pt}{x} \nput{0}{x}{$x^*$}
%\cnode*(1,-0.2){2.5pt}{y} \nput[labelsep=0pt]{150}{y}{\psframebox[fillstyle=solid,fillcolor=lightgray,framesep=1pt,linestyle=none]{$y^*$}}
%\rput[l](-0.95,0.8){$[-1,1]^n$}
%\ncline[arrowsize=6pt,linewidth=1pt]{<->}{x}{y}
%\end{pspicture}
\end{center}
Note that the step (2) is a convex program which can be solved in polynomial time, see \cite{DBLP:books/sp/GLS1988}.
Now we can finish the proof of Theorem~\ref{thm:PartialColoringForExpSmallSets}.
\begin{lemma} \label{lem:PartialColoringForExpSmallSetsMainLemma}
  If $\varepsilon,\delta > 0$ are chosen small enough (depending on $\alpha$), then
  with probability $1- e^{-\Omega_{\varepsilon,\delta}(n)}$ one has  $|\{ i \in [n] : y^*_i \in \{ -L_i,R_i\} \}| \geq \delta n$.
\end{lemma}

\begin{proof}
  For a set of indices $I \subseteq [n]$ we abbreviate the subspace $H(I) := \{ \bm{x} \in H \mid x_i =0 \; \forall i \in I\}$.
  Moreover we abbreviate $K(I) := \{ \bm{x} \in K \mid -L_i \leq x_i \leq R_i \; \forall i \in I\}$ as the intersection of $K$ with the slabs corresponding to coordinates in $I$.
  % let $K(I) := K \cap H(I)$ be the intersection of $K$ with the coordinate subspaces belonging to coordinates in $I$.
  Consider the two events
  \begin{eqnarray*}
    \pazocal{E}_1 &:=& ``d(\bm{x}^*,K \cap [-\bm{L},\bm{R}]) \geq (1-5\varepsilon) \cdot \sqrt{n}\text{''} \\
    \pazocal{E}_2 &:=& ``\textrm{for all }I \subseteq [n]\textrm{ with }|I| \leq \delta n\textrm{ one has }d(\bm{x}^*,K \cap H(I)) \leq (1-10\varepsilon) \sqrt{n}\text{''} 
  \end{eqnarray*}
We will see that both events $\pazocal{E}_1$ and $\pazocal{E}_2$ happen with overwhelming probability. \\
  {\bf Claim I.} \emph{One has $\Pr[\pazocal{E}_1] \geq 1 - \exp(-\frac{\varepsilon^2}{2}  n)$.} \\
  {\bf Proof of Claim I.} Follows from Lemma~\ref{lem:DistanceOfGaussianFromEpsilonCube} as $d(\bm{x}^*,K \cap [-\bm{L},\bm{R}]) \geq d(\bm{x}^*,K \cap [-\varepsilon,\varepsilon]^n) \geq d(\bm{x}^*,[-\varepsilon,\varepsilon]^n)$. \\
  {\bf Claim II.} \emph{If $\varepsilon,\delta >0$ are small enough, then $\Pr[\pazocal{E}_2] \geq 1-e^{-\Theta_{\varepsilon}(n)}$.} \\
  {\bf Proof of Claim II.} For any index set $I$ one can lower bound the measure as % index set $I$ with $|I| \leq \delta n$ one can lower bound the measure as
  $
  \gamma_{H(I)}(K \cap H(I)) \geq \gamma_H(K) \geq e^{-\alpha n}
$
 by Lemma ~\ref{lem:GaussianMeasureSections}.   
%  \[
%  \gamma(K(I)) \stackrel{\textrm{\v{S}idak-Kathri (Lem~\ref{lem:SidakLemma})}}{\geq} \gamma_n(K) \cdot \gamma_1([-\varepsilon,\varepsilon])^{|I|} \geq e^{-\alpha n} \cdot (\varepsilon/2)^{|I|} \geq e^{-\alpha n - \ln(\frac{2}{\varepsilon}) \cdot \delta n} \geq e^{-2\alpha n},
%\]
%assuming $\delta>0$ is chosen small enough so that $\ln(\frac{2}{\varepsilon}) \cdot \delta \leq \alpha$. 
%Here we use that $\gamma_1([-\varepsilon,\varepsilon]) \geq 2\varepsilon \cdot \gamma_1(1/2) \geq 2\varepsilon \frac{1}{\sqrt{2\pi}}e^{-(1/2)^2/2} \geq \frac{\varepsilon}{2}$ for $0<\varepsilon \leq \frac{1}{2}$.
 Let us abbreviate $\pazocal{I} := \{ I \subseteq [n] : |I| \leq \delta n\}$ as the family of small index sets. For $I \in \pazocal{I}$ we have
 $\dim(H(I)) \geq \dim(H) - |I| \geq \frac{\beta}{2} n$, if we choose  $\delta \leq \frac{\beta}{2}$.
 Then by Lemma~\ref{lem:DistanceGaussianToExpSmallSetDim} we 
know that a fixed $I \in \pazocal{I}$ has $\E_{\bm{x} \sim N(\bm{0},\bm{I}_n)}[d(\bm{x},K \cap H(I))] \leq \sqrt{n} \cdot \big(1-\frac{\beta/2}{512 \cdot \alpha e^{4\alpha}}\big) \leq (1-20\varepsilon) \sqrt{n}$, if we choose $\varepsilon \leq \frac{\beta/2}{20 \cdot 512\alpha e^{4\alpha}}$. Then by concentration one has $\Pr_{\bm{x} \sim N(\bm{0},\bm{I}_n)}[d(\bm{x},K \cap H(I)) > (1-10\varepsilon) \sqrt{n}] \leq \exp(-50\varepsilon^2n)$, see Theorem~\ref{thm:ConcentrationFor1LipschitzFunctions}. A useful bound is $|\pazocal{I}| \leq e^{2\delta\log_2(\frac{1}{\delta}) n} \leq e^{\varepsilon^2 n}$ if we choose $\delta$ small enough compared to $\varepsilon$. Then
\begin{eqnarray*}
  \Pr[\pazocal{E}_2] &\stackrel{\textrm{union bound}}{\leq}& \sum_{I \in \pazocal{I}} \Pr\big[d(\bm{x}^*,K \cap H(I)) > (1-10\varepsilon)\sqrt{n}\big] \\
  &\leq& e^{\varepsilon^2 n} \cdot \exp(-50\varepsilon^2n) \leq \exp\Big(-40\varepsilon^2n\Big). \qed
\end{eqnarray*}

Now we have everything to finish the proof. Fix an outcome of the vector $\bm{x}^*$ so that
the events $\pazocal{E}_1$ and $\pazocal{E}_2$ are both true, and abbreviate $I^* := \{ i \in [n] : y_i^* \in \{ -L_i,R_i\}\}$.
Suppose for the sake of contradiction that  $|I^*| < \delta n$. Then 
\begin{eqnarray*}
  (1-10\varepsilon)\sqrt{n}  &\stackrel{\pazocal{E}_2\textrm{ true } \& \; I^* \in \pazocal{I}}{\geq}& d(\bm{x}^*,K \cap H(I^*)) \\ &\stackrel{K \cap H(I^*) \subseteq K(I^*)}{\geq}& d(\bm{x}^*,K(I^*)) \\ &\stackrel{(*)}{=}& d(\bm{x}^*,K \cap [-\bm{L},\bm{R}]) \\ &\stackrel{\pazocal{E}_1\textrm{ true}}{\geq}& (1-5\varepsilon) \sqrt{n}
\end{eqnarray*}
which is a contradiction. Here the crucial argument for $(*)$ is that $d(\bm{x}^*,K \cap [-\bm{L},\bm{R}]) = \min\{ \|\bm{x}^*-\bm{y}\|_2 : \bm{y} \in K\textrm{ and }-L_i \leq y_i \leq R_i \; \forall i \in [n]\} $ is a \emph{convex minimization} problem and the optimum value will not change if linear constraints are discarded that are not tight for the optimum $\bm{y}^*$, and the box constraints for coordinates  $I^* \setminus [n]$ are indeed not tight.
\end{proof}

%  $(1-\beta)n$ many elements rather than $\delta n$ for some small constant $\delta$.
% It is not hard to color close to. As this is often
%We state that result which we claimed earlier in
We stated such a result earlier in Theorem~\ref{thm:PartialColoringForExpSmallSetsWithShiftAndLargeFractionColored}. Now we are ready to prove it:  
\begin{proof}[Proof of Theorem~\ref{thm:PartialColoringForExpSmallSetsWithShiftAndLargeFractionColored}]
  The basic idea is to simply apply Theorem~\ref{thm:PartialColoringForExpSmallSets} a constant number of times
  until the desired number of elements is colored. 
  We assume $\beta>\gamma$ since otherwise there is nothing to prove. Let $\varepsilon := \varepsilon(\alpha,\gamma),\delta := \delta(\varepsilon,\gamma) > 0$ be the constants from
  Theorem~\ref{thm:PartialColoringForExpSmallSets} that work for the given $\alpha$ and $\beta' := \gamma > 0$.
  
  We set $\bm{y}^{(0)} := \bm{y}$ and
  for $t \geq 0$ we set $F^{(t)} := \{ i \in [n] : y^{(t)}_i \in \{ -1,1\} \}$ as the variables that are \emph{frozen}.
  Suppose for some $t$ we have constructed a sequence $\bm{y}^{(0)},\ldots,\bm{y}^{(t)}$ and still $|F^{(t)}| < (\beta-\gamma)n$.
  Set $H^{(t)} := \{ \bm{x} \in H \mid x_i = 0 \; \forall i \in F^{(t)} \}$ be the subspace of $H$ where we fix frozen coordinates to be $0$.
  Note that $\dim(H^{(t)}) \geq \dim(H) - |F^{(t)}| \geq \gamma n$. Moreover $\gamma_{H^{(t)}}(K \cap H^{(t)}) \geq \gamma_H(K) \geq e^{-\alpha n}$ by Lemma~\ref{lem:GaussianMeasureSections}.
  We set $R_i := \frac{\varepsilon}{2} \cdot (1-y_i^{(t)})$ and $L_i := \frac{\varepsilon}{2} \cdot (y_i^{(t)} - (-1))$ for $i \in [n] \setminus F^{(t)}$
  and $R_i := L_i := \varepsilon$ for $i \in F^{(t)}$ and apply Theorem~\ref{thm:PartialColoringForExpSmallSets}.
  With high probability, the algorithm succeeds and provides a vector $\bm{x}^{(t)}$. We update $\bm{y}^{(t+1)} := \bm{y}^{(t)} + \frac{2}{\varepsilon} \bm{x}^{(t)} \in [-1,1]^n$
  where $\|\bm{y}^{(t+1)}\|_K \leq \|\bm{y}^{(t)}\|_K + \frac{2}{\varepsilon}$ by the triangle inequality. Moreover, the number of frozen coordinates increases\footnote{For frozen coordinates $i$ we did set $L_i = R_i = \varepsilon$ so that $\bm{x}^{(t)}$ will indeed contain $\delta n$ ``fresh'' coordinates that become tight, rather than rediscovering the coordinates in $F^{(t)}$.} to
  $|F^{(t+1)}| \geq |F^{(t)}| + \delta n$. We will terminate after at most $\frac{1}{\delta}$ iterations and if $T$ is the final iteration,  then $\bm{y}^{(T)} \in [-1,1]^n \cap \frac{2}{\varepsilon \delta} K$ as desired.
%  Let $K_{S^{(t)}} := \{ \bar{\bm{x}} \in \setR^{S^{(t)}}: (\bar{\bm{x}},\bm{0}) \in K\}$ and note that %\rem{V: Seems like $e^{-\alpha|S|} \ge e^{-\alpha n}$. \\ T: Good catch. I fixed it.}
%  $\gamma_{|S^{(t)}|}(K_{S^{(t)}}) \geq \gamma_n(K) \geq e^{-\alpha n} \geq \exp(- \frac{\alpha}{1-\beta+\gamma} |S^{(t)}|)$. Moreover $\dim(H_{S^{(t)}}) \geq \dim(H) - (n-|S^{(t)}|) \geq \beta n - (\beta-\gamma)n = \gamma n$. Hence by Theorem~\ref{thm:PartialColoringForExpSmallSetsWithShift} there exists a $\bm{x}^{(t)}$ so that $\bm{y}^{(t+1)} := \bm{y}^{(t)} + \bm{x}^{(t)} \in [-1,1]^n$ with $\bm{x}^{(t)} \in (C'  \cdot K \cap H)$ and $|S^{(t+1)}| \leq (1-\delta) |S^{(t)}|$ for some constants $C',\delta>0$. As soon as we reach an iteration $t$ with $|S^{(t)}| < (1-\beta+\gamma)n$ we stop and return the desired vector $\bm{x} := \bm{x}^{(0)} + \ldots + \bm{x}^{(t-1)}$.
\end{proof}

We would like to mention that Theorem~\ref{thm:PartialColoringForExpSmallSetsWithShiftAndLargeFractionColored} may also be deduced, after some work, from the Gaussian measure amplification techniques derived in \cite{8555088} with the use of $\alpha$-regular M-ellipsoids. We believe the analysis presented here is simpler, since the existence of such regular M-ellipsoids is a deep result in convex geometry. 

\section{From hereditary volume bounds to Gaussian measure}

This section is devoted to the proof of Theorem~\ref{thm:HereditaryVolumeImpliesGaussianMeasure}, which provides a connection between hereditary volume and Gaussian measure. For a brief motivation, note that for any convex body $K \subseteq \setR^n$ and any $S \subseteq [n]$ one has $\vol_{|S|} (K_S) \ge \gamma_{|S|} (K_S) \ge \gamma_n (K)$. It is therefore a natural question whether a converse holds, and Theorem~\ref{thm:HereditaryVolumeImpliesGaussianMeasure} shows that this is indeed the case. As a corollary, we settle up to an exponential factor a conjecture of ~\cite{GeometryOfLpBall-BartheGuedonMendelsonNaor2005} that coordinate sections minimize the Gaussian measure among all sections of scaled $\ell_p$ balls.

We would also like to mention that we cannot hope for a refinement of the right side to only sections of dimension $\delta n$. For example when $K = \eps \cdot B^{\delta n -1}_2 \times \setR^{n-\delta n + 1}$, all $\delta n$-dimensional sections have infinite volume yet $\gamma(K) \to 0$ as $\eps \to 0$.

 While relatively short, our proof does use several auxilliary results. The key ingredient is the following formula which expresses the volume of the Minkowski sum of a convex body and an Euclidean ball as a weighted sum of \textit{quermassintegrals} $W_i(K)$ which are average volumes of projections. Recall that given $A,B\subseteq \setR^n$, $A+B:= \{\bm{a}+\bm{b}: \bm{a} \in A, \bm{b} \in B\}.$ 

 %\rem{T: It seems $W_i(K)$ is typically used for the Quermassintegrals and $V_j(K)$ are the intrinsic volumes (which apparently are a different normalization of the former. At least according to \url{https://en.wikipedia.org/wiki/Mixed_volume} the current definition is a bit of a hybrid between instrinsic volume and quermassintegral.}

\begin{lemma}[Kubota's Integral Formula ~\cite{pisier_1989}] \label{lem:KubotaIntegralFormula}
For any convex body $K \subset \setR^n$, we have
\[
\vol_n(K + \lambda B^n_2) = \sum_{i=0}^n \lambda^{i} {n \choose i} W_i(K)
\]
with
\[
W_{n-i}(K) := \frac{\vol_n(B^n_2)}{\vol_i(B^i_2)} \int_{G(n,i)} \vol_i(\pi_L(K)) dL, 
\]
where the integral is over the uniform measure over  $G(n,i)$, which is the set of $i$-dimensional linear subspaces $L \subseteq \setR^n$ and $\pi_L(K)$ denotes the orthogonal projection of $K$ onto $L$.
\end{lemma}%\rem{T: In some sense the normalization of $\int_{G(n,i)}.. dL$ was not specified. } \rem{V: I think "uniform measure" already implies $\int 1 = 1$. But it's fine to add a clarification if you see fit.}

In order to relate projections to slices, we use polarity. Given a symmetric convex set $K \subseteq \setR^n$, its \textit{polar} is $K^\circ := \{\bm{y} \in \mathrm{span}(K) \mid \inn{\bm{x}}{\bm{y}} \le 1 \ \forall \bm{x} \in K\}$. The following lemma elucidates the reason polars are helpful to transform projections into slices:

\begin{lemma} \label{lem:ProjectionPolar}
Given a symmetric convex body $K \subseteq \setR^n$ and any subspace $H \subseteq \setR^n$, we have $(K \cap H)^\circ = \pi_H(K^\circ)$.
\end{lemma}

%\rem{V: I had to use psellipse vs psellipticwedge to get the same picture, shouldn't matter much.}
\begin{center}
\psset{unit=1.3cm}
\begin{pspicture}(-2,-1.85)(3,1.85)
%\rput{12}(0,0){\psellipticwedge[fillstyle=solid,fillcolor=lightgray](0,0)(1.5,0.5){0}{-1}} \rput[c](1.2,0.3){$K$}
\rput{12}(0,0){\psellipse[fillstyle=solid,fillcolor=lightgray](0,0)(1.5,0.5)} \rput[c](1.2,0.3){$K$}
\psline[linestyle=dashed,linewidth=0.5pt](-2,0)(2,0)
\pscircle[linestyle=dashed,linewidth=0.5pt](0,0){1}
\psline[linewidth=1.5pt](-1.25,0)(1.25,0)
%\cnode*(1.2,0){2.5pt}{a} \nput[labelsep=2pt]{-60}{a}{$\bm{u}$}
\cnode*(0,0){2.5pt}{origin}\nput[labelsep=2pt]{135}{origin}{$\bm{0}$}
\rput[c](-1.8,5pt){$H$}
\rput(1.35;-135){$B_2^n$}
\pnode(0.5,0){A}
\pnode(0.5,1.3){B}
\ncline{->}{B}{A}
\nput[labelsep=2pt]{90}{B}{$K \cap H$}
%\rput[c](0.5,7pt){$K \cap H$}
\end{pspicture}
\begin{pspicture}(-2,-1.85)(2,1.85)
%\rput{12}(0,0){\psellipticwedge[fillstyle=solid,fillcolor=lightgray](0,0)(0.666,2){0}{-1}} \rput[c](-0.1,1.3){$K^{\circ}$}
\rput{12}(0,0){\psellipse[fillstyle=solid,fillcolor=lightgray](0,0)(0.666,2)} \rput[c](-0.1,1.3){$K^{\circ}$}
\psline[linestyle=dashed,linewidth=0.5pt](-2,0)(2,0)
\pscircle[linestyle=dashed,linewidth=0.5pt](0,0){1}
\psline[linewidth=1.5pt](-0.78,0)(0.78,0)
\psline[linestyle=dotted](0.78,-2)(0.78,2)
\psline[linestyle=dotted](-0.78,-2)(-0.78,2)
%\cnode*(0.8333,0){2.5pt}{a} \nput[labelsep=2pt]{-45}{a}{$\frac{\bm{u}}{\|\bm{u}\|_2^2}$}
\cnode*(0,0){2.5pt}{origin}\nput[labelsep=2pt]{135}{origin}{$\bm{0}$}
\rput[c](-1.8,5pt){$H$}
%\rput[l](0.9,1.5){$\left<\bm{u},\bm{x} \right> \leq 1$}
\rput(1.35;-150){$B_2^n$}
\pnode(0.4,0){A}
\pnode(1.5,1.3){B}
\ncline{->}{B}{A}
\nput[labelsep=2pt]{90}{B}{$\pi_{H}(K^{\circ})$}
\end{pspicture}
\end{center}
It is also well-known that polarity transforms intersections into convex hulls:

\begin{lemma} \label{lem:PolarElementary}
Given symmetric convex bodies $K, L \subseteq \setR^n$, we have $(K \cap L)^\circ = \conv(K^\circ,L^\circ)$.

%(a) $K \subseteq L \implies L^\circ \subseteq K^\circ$.

%(b) $K = (K^\circ)^\circ$. 

\end{lemma}
For a detailed introduction to polarity we refer to Rockefellar~\cite{Rockafellar1970ConvexAnalysis}.
Finally, we need the Blaschke-Santal\'o Inequality and its deep converse due to Bourgain-Milman ~\cite{AsymptoticGeometricAnalysisBook2005}:

\begin{lemma} \label{lem:SantaloInequality}
Given a symmetric convex body $K \subseteq \setR^n$, we have
 $2^{O(n)} \ge \frac{\vol_n(K) \cdot \vol_n(K^\circ)}{\vol_n(B^n_2)^2} \ge 2^{-O(n)}.$ 
\end{lemma}

%{}\rem{V: Do you know a good reference for Lemmas 27-29? \\ T: I put Rockafellar. Not my favourite type of book but he spells out a lot that AMG don't bother justifying.}

The starting point of the proof, which connects the Gaussian measure to the Minkowski sum with an Euclidean ball, is given by the following bound:

%\rem{T: Isn't there a factor of 1/2 missing. This looks like you use $conv(A,B) \supseteq A+B$ which isn't true. But $conv(A,B) \supseteq \frac{1}{2}(A+B)$ is fine.}
%\rem{V: Actually I was using the other direction, but it was rather confusing so I made it more clear.}

\begin{lemma} \label{lem:PolarGaussianLowerBound}
Given a symmetric convex body $K \subseteq \setR^n$, $\gamma_n(K) \ge \vol_n \Big(K^\circ + \frac{1}{\sqrt{n}} B^n_2\Big)^{-1} \cdot n^{-n} \cdot 2^{O(n)}$.
\end{lemma}

\begin{proof}
We start by noting that we can lower bound the Gaussian measure upon restriction to a $\sqrt{n}$-radius ball:
\[
\gamma_n(K) = \frac{1}{(2\pi)^{n/2}} \int_K e^{-\|\bm{x}\|_2^2/2} \ \mathrm{d}\bm{x} \ge \frac{1}{(2\pi e)^{n/2}} \vol_n(K \cap \sqrt{n} B^n_2),
\]
and since $(K \cap \sqrt{n} B^n_2)^\circ = \conv(K^\circ, \frac{1}{\sqrt{n}} B^n_2)$ by Lemma~\ref{lem:PolarElementary}, we conclude 
\begin{eqnarray*}
\gamma_n(K) & \ge & \vol_n (K \cap \sqrt{n} B^n_2) \cdot 2^{-O(n)} \\ & \stackrel{\textrm{Lem} ~\ref{lem:SantaloInequality}}{\ge} & \vol_n\Big(\conv\Big(K^\circ, \frac{1}{\sqrt{n}} B^n_2\Big)\Big)^{-1} \cdot n^{-n} \cdot 2^{-O(n)} \\ & \ge & \vol_n \Big(K^\circ + \frac{1}{\sqrt{n}} B^n_2\Big)^{-1} \cdot n^{-n} \cdot 2^{O(n)},
\end{eqnarray*}
since $\conv(K^\circ, \frac{1}{\sqrt{n}} B^n_2) \subseteq  K^\circ + \frac{1}{\sqrt{n}} B^n_2$. 
\end{proof}

In order to connect slices to \textit{coordinate} slices, we apply a result of~\cite{8555088} for ellipsoids. Thus we will need to use the existence of M-ellipsoids~\cite{AsymptoticGeometricAnalysisBook2005}:

\begin{lemma} \label{lem:ExistenceOfMEllipsoids} 
For any symmetric convex body $K \subseteq \setR^n$ there exists an ellipsoid $E \subseteq \setR^n$ for which there exist collections of centers $S_E, S_K \subseteq \setR^n$ with $|S_E|,|S_K| \le 2^{O(n)}$ so that $K \subseteq \bigcup_{c \in S_E} (c+E)$ and $E \subseteq \bigcup_{c' \in S_K} (c'+K)$.
\end{lemma}

\begin{proof}[Proof of the first inequality in Theorem~\ref{thm:HereditaryVolumeImpliesGaussianMeasure}]

Kubota's integral formula (Lemma ~\ref{lem:KubotaIntegralFormula}) applied to $K^\circ$ yields
\[
W_{n-i} (K^\circ) = \frac{\vol_n(B^n_2)}{\vol_i(B^i_2)} \int_{G(n,i)} \vol_i(\pi_L (K^\circ)) dL.
\]
By Lemma~\ref{lem:ProjectionPolar} and Santal\'o's inequality (Lemma ~\ref{lem:SantaloInequality}) we know that for any subspace $L$,
\[
\vol_i(\pi_L(K^\circ)) \le \vol_i (B^i_2)^2 \cdot \vol_i (K \cap L)^{-1} \le M^{-1} \cdot i^{-i} \cdot 2^{O(i)},
\]
where we choose to denote $\displaystyle M := \min_{\dim L = i \le n} \vol_i (K \cap L)$.   We conclude

\[W_{n-i} (K^\circ) \le M^{-1} \cdot n^{-n/2} \cdot i^{-i/2} \cdot 2^{O(n)},\]
so that 
\[ n^{-(n-i)/2} \cdot W_{n-i} (K^\circ) \le M^{-1} \cdot n^{-n} \cdot 2^{O(n)},\]
by using $(n/i)^i \le 2^{O(n)}$ for $i \in [n]$. Taking $\lambda := 1/\sqrt{n}$ and summing over $i \in [n]$ in Lemma~\ref{lem:KubotaIntegralFormula} gives

\[
\vol_n \Big(K^\circ + \frac{1}{\sqrt{n}} B^n_2\Big) \le M^{-1} \cdot n^{-n} \cdot 2^{O(n)},
\]
so that by Lemma~\ref{lem:PolarGaussianLowerBound} we obtain $\gamma_n(K) \ge M \cdot 2^{-O(n)}$. It remains to show that the minimal \textit{coordinate} sections are not much larger than the minimal sections. With this purpose in mind, let $E$ be an M-ellipsoid of $K$. By Lemma~\ref{lem:ExistenceOfMEllipsoids}, there exist collections $S_E, S_K$ with $|S_E|,|S_K| \le 2^{O(n)}$ so that $K \subseteq \bigcup_{c \in S_E} (c+E)$ and $E \subseteq \bigcup_{c' \in S_K} (c'+K)$. Note that for any $i$-dimensional subspace $L$ we have
\[
\vol_i (K \cap L) \le \sum_{c \in S_E} \vol_i ((c+E) \cap L) \le 2^{O(n)} \cdot \vol_i (E \cap L)
\]
and similarly 
\[
\vol_i (E \cap L) \le \sum_{c' \in S_K} \vol_i ((c'+K) \cap L) \le 2^{O(n)} \cdot \vol_i (K \cap L),
\]
where by Brunn's concavity principle the sections with largest volume are those through the origin. Thus it suffices to show that
\[
\min_{\dim L = i} \vol_i (E \cap L) \ge \min_{S \subseteq [n], |S| = i} \vol_i (E_S) \cdot 2^{-O(n)}.
\]
Indeed this follows a form of restricted invertibility in the work of Dadush, Nikolov, Talwar and Tomczak-Jaegermann, who showed in ~\cite{8555088} (see p. 8) an improved bound of

\[
\min_{\dim L = i} \vol_i (E \cap L) \ge \min_{S \subseteq [n], |S| = i} \vol_i (E_S) \cdot {n \choose i}^{-1}. \qedhere
\]
\end{proof}

We now prove the second part of Theorem~\ref{thm:HereditaryVolumeImpliesGaussianMeasure} which restricts our attention to sections of dimension $\le \delta n$. For this we need the following inequality for quermassintegrals which can be seen as a strenghtening of the isoperimetric inequality:
{}
\begin{theorem}[Alexandrov Inequality~\cite{pisier_1989}] \label{lem:AlexandrovInequality}
Given $i \ge j$ we have 

\[
\Big(\frac{W_{n-i}(K)}{\vol_{i}(B^{i}_2)}\Big)^{1/i} \le \Big(\frac{W_{n-j}(K)}{\vol_{j}(B^{j}_2)}\Big)^{1/j}.
\]
\end{theorem} 	

\begin{proof}[Proof of the second inequality in Theorem~\ref{thm:HereditaryVolumeImpliesGaussianMeasure}]
We proceed as in the proof of the first inequality. Setting $\lambda := 1/\sqrt{n}$ we still have, for $j \le \delta n$, 
\begin{align*}
\lambda^{n-j} W_{n-j}(K^\circ) & \le \max_{\dim L = i \le \delta n} \vol_i^{-1} (K \cap L) \cdot n^{-n} \cdot 2^{O(n)} \\ & \le \max_{\dim L = i \le \delta n} \vol_i^{-1/\delta} (K \cap L) \cdot n^{-n} \cdot 2^{O(n)},
\end{align*}
as the maximum is at least one (for $i = 0$). For $j > \delta n$ we use Theorem~\ref{lem:AlexandrovInequality} to see that
\[\lambda^{n-j} W_{n-j}(K^\circ) \le \lambda^{n-j} (W_{n-\delta n} (K^\circ))^{j/(\delta n)} \cdot \vol_j(B^j_2) \cdot \vol_{\delta n} (B^{\delta n}_2)^{-j/\delta n}\]
and proceed as in the first half of the proof: %\rem{T: In the 3 $\max$ expressions, wouldn't it be $i>\delta n$?}
\begin{align*}
 \lambda^{n-j} W_{n-j}(K^\circ) & \le \lambda^{n-j} \cdot (W_{n-\delta n} (K^\circ))^{j/(\delta n)} \cdot \vol_j(B^j_2) \cdot \vol_{\delta n} (B^{\delta n}_2)^{-j/\delta n} \\
 & \le \lambda^{n-j} \cdot \Big(\max_{\dim L = i \le \delta n} \vol_i^{-1} (K \cap L) \cdot n^{-n/2} \cdot (\delta n)^{-\delta n/2} \Big)^{j/(\delta n)} \cdot (\delta n/j)^{j/2} \cdot 2^{O(n/\delta)} \\
  & \le  \lambda^{n-j}  \cdot n^{-j/(2\delta)} \cdot j^{-j/2} \cdot \max_{\dim L = i \le \delta n} \vol_i^{-1/\delta} (K \cap L) \cdot 2^{O(n/\delta)}  \\
   & = n^{-n/2} \cdot n^{-j/(2\delta)} \cdot \underbrace{(n/j)^{j/2}}_{\le 2^{O(n)}} \max_{\dim L = i \le \delta n} \vol_i^{-1/\delta} (K \cap L) \cdot 2^{O(n/\delta)} \\
    & \le n^{-n} \cdot \max_{\dim L = i \le \delta n} \vol_i^{-1/\delta} (K \cap L) \cdot 2^{O(n/\delta)}. 
\end{align*}
The statement follows as before: by summing over $j \in [n]$ in Lemma~\ref{lem:KubotaIntegralFormula} we obtain
\[
\gamma_n(K) \ge \min_{\dim L = i \le \delta n} \vol_i^{1/\delta} (K \cap L) \cdot 2^{-O(n/\delta)},
\]
and we can pass to coordinate sections via M-ellipsoids.%\rem{T: I added the dimensions to all the $\vol$.}
\end{proof}

\begin{remark}
Barthe, Gu\'edon, Mendelson, and Naor conjectured that coordinate slices maximize the Gaussian volume among all slices of a (scaled) $\ell_p$ ball ~\cite{GeometryOfLpBall-BartheGuedonMendelsonNaor2005} (see the remark in p. 28). We can use the above result to give an affirmative answer up to $2^{-O(n)}$:

\begin{corollary}
Let $p \ge 2$, $r > 0$ and $H \subseteq \setR^m$ an $n$-dimensional subspace. Then 
\[\gamma_H (r B^m_p \cap H) \ge \gamma_n(rB^n_p) \cdot 2^{-O(n)}.\] 
\end{corollary}
\begin{proof}
If $r > n^{1/p}$, the right side is already $2^{-O(n)}$ so we may assume that $r \le n^{1/p}$. A well-known result of Meyer-Pajor asserts that coordinate sections minimize the \textit{volume} among all sections of the $\ell_p$ ball ~\cite{MEYER1988109}. Applying Theorem~\ref{thm:HereditaryVolumeImpliesGaussianMeasure} and using Meyer-Pajor we get 
\begin{align*} 
\gamma_H (r B^m_p \cap H) \ge & \min_{L \subseteq H, \dim L = i} \vol_i (rB^m_p \cap L) \ge \min_{i \le n} \vol_i (rB^i_p) \ge \gamma_n(rB^n_p) \cdot 2^{-O(n)}. \qedhere{}
\end{align*}
\end{proof}
\end{remark}

\begin{remark}
We mention another application of Theorem~\ref{thm:HereditaryVolumeImpliesGaussianMeasure}. For a symmetric convex $K \subseteq \setR^n$, denote the \textit{hereditary discrepancy} $\mathrm{hd}(K)$ as the minimum $t \ge 0$ so that $tK_S$ intersects $\{-1,1\}^S \times \{0\}^{[n]\setminus S}$ for all $S \subseteq [n]$. In ~\cite{8555088} it is shown that we have a lower bound $\mathrm{hd}(K) \ge \max_{S \subseteq [n]} \inf \{t : \vol_{|S|} (tK_S) \ge 1\},$ where the left side is known as the \textit{volume lower bound} $\mathrm{volLB}(K)$. In fact an analogous argument also shows the lower bound $\mathrm{hd}(K) \ge \max_{S \subseteq [n]} \inf \{t : \gamma_{|S|} (tK_S) \ge 2^{-C|S|}\}$ for a universal constant $C > 0$. Since the volume of a convex body is always lower bounded by its Gaussian measure, this lower bound is at least $\mathrm{volLB}(K)$ up to a factor of $2^C$. Theorem~\ref{thm:HereditaryVolumeImpliesGaussianMeasure} immediately implies that it is also at most $\mathrm{volLB}(K)$ up to a constant. 

\end{remark}

\section{Open problems}
We conjecture that Theorem~\ref{thm:FullColoringForLpLq} can be improved to match Theorem~\ref{thm:PartialColoringForLpLq}:

\begin{conjecture}[$\ell_p \to \ell_q$ version of Koml\'os conjecture] \label{con:LpLqKomlos}
Given $n \le m$, $2 \le p \le q \le \infty$ and $\bm{a}_1, \dots, \bm{a}_n \in B^m_p$, do there always exist signs $\bm{x} \in \{-1,1\}^n$ so that
\[
\displaystyle \Big\|\sum_{i=1}^n x_i \bm{a}_i\Big\|_q \le C \sqrt{\min\Big(p,\log\Big(\frac{2m}{n}\Big)\Big)}\cdot n^{1/2 - 1/p + 1/q},
\]
for some universal constant $C > 0$?
\end{conjecture}

Since Conjecture~\ref{con:LpLqKomlos} is at least as hard as the Koml\'os conjecture, a more realistic goal would be to improve the full coloring of Theorem~\ref{thm:FullColoringForLpLq} by a factor of $(1/2 - 1/p + 1/q)^{-1/2}$ so as to match the best known bound of $O(\sqrt{\log n})$ for Koml\'os.

Recall that for a matrix $\bm{A} \in \setR^{n \times n}$ and $1 \leq p \leq \infty$,
the \emph{Schatten-$p$ norm} is defined as $\|\bm{A}\|_{S(p)} := (\sum_{i=1}^n \sigma_i(\bm{A})^p)^{1/p}$ where $\sigma_i(\bm{A}) \geq 0$ is the $i$th \emph{singular value} of the matrix. In particular $\|\bm{A}\|_{S(\infty)}$ is the maximum singular value and $\|\bm{A}\|_{S(1)}$ is known as \emph{Trace norm} or \emph{Nuclear norm}.
One might wonder whether Theorem~\ref{thm:PartialColoringForLpLq} could be extended for \emph{matrices} instead of vectors in the corresponding Schatten norms. In fact this is not possible: even for $p = 2$ and $q = \infty$, there exist $n$ rank-one matrices $\bm{A}_i := \bm{v}_i \bm{v}_i^\top \in \setR^{n \times n}$ with unit $\bm{v}_i$ for which any fractional coloring has discrepancy $\Omega(\sqrt{n})$ in the operator norm (\cite{weaver2002}, Section 3). It is still possible nevertheless that Corollary~\ref{cor:LpSpencer} extends in the following way:
%\newpage
\begin{conjecture}[$\ell_p$ version of Matrix Spencer] \label{con:LpMatrixSpencer} Given $2 \le p \le \infty$ and symmetric $\bm{A}_1, \dots, \bm{A}_n \in \setR^{n \times n}$ with Schatten-$p$ norm at most 1, can we always find signs $\bm{x} \in \{-1,1\}^n$ so that

\[
\displaystyle \Big\|\sum_{i=1}^n x_i \bm{A}_i \Big\|_{S(p)} \le C\sqrt{n}
\] 
for some universal constant $C > 0$?
\end{conjecture}
This is a more general form of the Matrix Spencer conjecture \cite{MatrixBalancingZouziasICALP12}, and one can show a weaker bound of $O(\sqrt{pn})$ with random signs similar to Lemma~\ref{lem:ExpectedLpNorm}. In fact, it is an open problem to show even a partial coloring for Conjecture 2. This would be implied by the following, which at least holds for diagonal matrices by the proof of Lemma~\ref{lem:CaseqInfty}:
\begin{conjecture} \label{con:MatrixSpencerMeasureBound}
Given $1 \le p \le \infty$ and symmetric $\bm{A}_1, \dots, \bm{A}_n \in \setR^{n \times n}$, can we show that
\[
K := \Big\{\bm{x} \in \setR^{n} : \Big\|\sum_{i=1}^n x_i \bm{A}_i\Big\|_{S(p)} \le \Big\|\Big(\sum_{i=1}^n \bm{A}_i^2\Big)^{1/2}\Big\|_{S(p)}\Big\}
\]
satisfies $\gamma_n(K) \ge 2^{-O(n)}$?
\end{conjecture}

%\noindent {\Large \textbf{Acknowledgments}}
%\\
%\\
\subsection*{Acknowledgments}
We would like to thank Daniel Dadush and Aleksandar Nikolov for their feedback in early drafts of this work and helpful discussions, and the anonymous reviewers for their detailed comments.

\bibliographystyle{alpha}
\bibliography{measureOfLpBall}

\appendix

\section{Proof of Lemma ~\ref{lem:ExpectedLpNorm}}

\begin{proof}[Proof of Lemma ~\ref{lem:ExpectedLpNorm}]
By convexity of $z \mapsto |z|^p$, Jensen's inequality in $(*)$ and Khintchine's inequality in $(**)$ (Lemma~\ref{lem:Khintchine}) we have 
\begin{align*}
\E \Big[\Big\|\sum_{i=1}^n x_i \bm{a}_i \Big\|_p\Big]  & \stackrel{(*)}{\le} \E \Big[ \Big\|\sum_{i=1}^n x_i \bm{a}_i \Big\|^p_p\Big]^{1/p} \\ & = \Big(\sum_{j \in [m]} \E \Big[ \Big|\sum_{i \in [n]} x_i a_{ij}\Big|^p \Big]\Big)^{1/p} \\ & \stackrel{(**)}{\le} C \sqrt{p} \cdot \Big(\sum_{j \in [m]} \Big(\sum_{i \in [n]} a^2_{ij}\Big)^{p/2}\Big)^{1/p}.
\end{align*}
If $p \in [1,2]$, write $\bm{A}_j \in \setR^n$ as $(\bm{A}_j)_i := a_{ij}$. Then by Lemma~\ref{lem:ElementaryInequalityOnLpNorms},
\[
\Big(\sum_{j \in [m]} \Big(\sum_{i \in [n]} a^2_{ij}\Big)^{p/2}\Big)^{1/p} = \Big(\sum_{j \in [m]} \|\bm{A}_j\|_2^p \Big)^{1/p}\le \Big(\sum_{j \in [m]} \|\bm{A}_j\|_p^p \Big)^{1/p} = \Big(\sum_{i \in [n]} \|\bm{a}_i\|_p^p \Big)^{1/p} \le n^{1/p}.
\]
Now suppose that $p \ge 2$. Define $(\bm{a}_i)^2 \in \setR^m$ to be the vector with $j$th coordinate $a^2_{ij}$. Since $\| \cdot \|_{p/2}$ is a norm, we can use the triangle inequality to get 
\[
\Big(\sum_{j \in [m]} \Big(\sum_{i \in [n]} a^2_{ij}\Big)^{p/2}\Big)^{1/p} = \Big\|\sum_{i \in [n]} (\bm{a}_i)^2\Big\|^{1/2}_{p/2} \le \Big(\sum_{i \in [n]} \|(\bm{a}_i)^2\|_{p/2}\Big)^{1/2} = \Big(\sum_{i \in [n]} \|\bm{a}_i\|^2_p\Big)^{1/2} \le n^{1/2}.
\]

Either way, we conclude that $\E [\|\sum_{i=1}^n x_i \bm{a}_i \|_p] \le O(\sqrt{p} \cdot n^{\max(1/2,1/p)})$, as desired.
\end{proof}

\begin{remark}
A similar approach gives an alternate proof of Prop. 25 in ~\cite{GeometryOfLpBall-BartheGuedonMendelsonNaor2005}, which states that a $r := O(\sqrt{p} \cdot n^{1/p})$ scaling of an $n$-dimensional section $H$ of $B^m_p$ has Gaussian measure $\gamma_H (H \cap rB^m_p) \ge 1/2$ for $p \ge 2$. Indeed, by Markov's inequality, it suffices to note that given an orthonormal basis $\bm{a}_1, \dots, \bm{a}_n$ of $H$ we have  
\[
\E \Big[\Big\|\sum_{i=1}^n x_i \bm{a}_i \Big\|_p\Big] \le C \sqrt{p} \cdot \Big(\sum_{j \in [m]} \Big(\sum_{i \in [n]} a^2_{ij}\Big)^{p/2}\Big)^{1/p} \le C \sqrt{p} \cdot n^{1/p},
\]
where the last inequality follows from convexity of $z \mapsto z^{p/2}$ and from the fact that the $m$ terms $\sum_{i \in [n]} a^2_{ij}$ sum to $n$ and are at most $1$ by orthonormality. 
\end{remark}

\section{Large convex sets without partial colorings\label{appendix:SetsWithoutPartialColorings}}

We have mentioned earlier that a symmetric convex set $K$ with measure $\gamma_n(K) \geq e^{-\delta n}$
contains a partial coloring $\bm{x} \in \{ -1,0,1\}^n$ with a linear number of nonzero coordinates if the constant $\delta$ is small enough --- but we claimed that this is false for constants beyond a certain threshold, even if one is allowed to rescale the
  body by some parameter dependent on $\delta$. The construction for such a set is a thin strip that avoids
  any point in $\{ -1,0,1\}^n \setminus \{ \bm{0}\}$.
\begin{lemma}
  For any $C \geq 1$, there exists a $\delta > 0$ so that the following holds: for any $n \in \mathbb{N}$ large enough
  there is a symmetric convex body $K \subseteq \setR^n$ so that (i) $(C^nK) \cap (\{ -1,0,1\}^n \setminus \{ \bm{0} \}) = \emptyset$ and (ii) $\gamma_n(K) \geq e^{-\delta n}$.
\end{lemma}
\begin{proof}
  The construction is probabilistic. 
  We sample a Gaussian $\bm{g} \sim N(\bm{0},\bm{I}_n)$ and for a tiny parameter $s > 0$ that we determine later,
we consider the strip $K := \{ \bm{x} \in \mathbb{R}^n : |\left<\bm{g},\bm{x}\right>| \leq s\}$.
Consider the set of nontrivial partial colorings $X := \{ -1,0,1\}^n \setminus \{ \bm{0} \}$ and recall that $|X| \leq 3^n$.
For any $\bm{x} \in X$, the distribution of $\left<\bm{g},\bm{x}\right>$ is Gaussian
with variance $\|\bm{x}\|_2^2 \geq 1$ and hence the density of this 1-dimensional Gaussian is at most $\frac{1}{\sqrt{2\pi}}e^0 \leq \frac{1}{2}$ everywhere.  In particular for a fixed $\bm{x} \in X$,
one can obtain the simple estimate of $\Pr[|\left<\bm{g},\bm{x}\right>| \leq t] \leq 4t$ for any $t>0$.
Then choosing  $s := \frac{1}{16} \cdot C^{-n} 3^{-n}$ we obtain
\[
  \Pr_{\bm{g}}\big[(C^nK) \cap X \neq \emptyset\big] \leq \sum_{\bm{x} \in X} \Pr_{\bm{g}}\big[|\left<\bm{g},\bm{x}\right>| > C^ns\big] \leq \frac{1}{4} \cdot |X| \cdot 3^{-n} \leq \frac{1}{4} \quad \quad (*)
\]
Moreover using Markov's Inequality we obtain the (rather weak) estimate
\[
   \Pr\big[\|\bm{g}\|_2^2 > 4n\big] \leq \frac{1}{4}  \quad \quad (**)
 \]
 Then with probability at least $1/2$ none of the events $(*)$ and $(**)$ happen. We fix such an outcome of $\bm{g}$
 and estimate that the measure of our strip is
 \[
  \gamma_n(K) = \int_{-s/\|\bm{g}\|_2}^{s/\|\bm{g}\|_2} \frac{1}{\sqrt{2\pi}}e^{-x^2/2}dx \geq  \frac{1}{\sqrt{2\pi}}e^{-1/2} \frac{2s}{\sqrt{n}} \geq e^{-\delta n}
 \]
for a suitable choice of $\delta$ using $\frac{s}{\|\bm{g}\|_2} \leq 1$.
\end{proof}

\end{document}